\documentclass[11pt]{article}
\usepackage{amsmath,amsbsy,amssymb,amsthm,constants,enumitem,citesort,cite} 
\usepackage{hyperref}
\usepackage{epsfig,epsf}
\usepackage{epstopdf}
\usepackage{bm}

\usepackage{graphicx,wrapfig,epsfig}
\usepackage{caption}
\usepackage{subcaption}
\usepackage{float}

\def\real    { \mathbb{R} }
\newtheorem{thm}{Theorem}
\newtheorem{lemma}{Lemma}

\newtheorem{prop}{Proposition}
\newtheorem{definition}{Definition}

\newcommand{\eps}{\epsilon}
\newcommand{\bitem}{\begin{itemize}}
\newcommand{\eitem}{\end{itemize}}

\newcommand{\argmin}{\mathrm{argmin}}

\newcommand{\beqn}{\begin{equation}}
\newcommand{\eeqn}{\end{equation}}
\newcommand{\balign}{\begin{align}}
\newcommand{\ealign}{\end{align}}

\newcommand{\notechecked}[1]{}

\def \CandesOne {C_1}
\def \CandesTwo {C_2}
\def \CddOne {C_3}
\def \CddTwo {C_4}
\def \CddThree {C_5}
\def \KappaOne {C_6}
\def \KappaTwo {C_7}
\def \GrimesOne {C_8}
\def \GrimesTwo {C_9}
\def \GrimesThree {C_{10}}

\def \curve { \gamma } 
\def \proj { \Phi } 
\def \noise { n } 
\def \dim {N}      
\def \pdim {M}     

\def \manifold { \mathcal{M} } 
\def \condition { \tau } 
\def \volume { V_\manifold } 
\def \gdist { {d_\manifold} } 

\def \sparsity {K} 
\def \mdim {K}     
\def \tan { \mathcal{T} } 
\def \ints {\mathbb{N}} 

\newcommand{\dist}[2]{\left\| #1 - #2 \right\|}
\newcommand{\norm}[1]{\left\| #1 \right\|} 

\newcommand{\pr}[1]{\operatorname{P}\left\{ #1 \right\}}
\newcommand{\vol}[2]{\mbox{vol}_{#1}(#2)}

\def \ball {\mathcal{B}} 
\def \intersection {\mathcal{A}} 
\def \thresh {\gamma} 
\def \Mresolution {\eta} 
\def \Tresolution {\eta'} 
\def \Uresolution {\delta} 
\def \coverno {\mathcal{N}} 
\def \packing {\mathcal{P}} 
\def \Ucover {T} 
\def \Ucoverno {\mbox{\textbf{t}}} 
\def \failprob {\rho} 
\def \Uparcover {L} 
\def \Uparcoverno {\operatorname{l}} 
\def \coverbnd{\operatorname{c}} 
\def \links{Q} 
\def \Vol{\mbox{vol}} 
\def \sing{\sigma} 
\def \dev{\delta} 

\def \uniti {\mbox{i}}

\title{{\bf New Analysis of Manifold Embeddings and Signal Recovery from Compressive Measurements}}

\author{Armin Eftekhari and Michael B. Wakin\footnote{Email: {aeftekha,mwakin}@mines.edu. This work was partially supported by NSF grant DMS-0603606, DARPA grant HR0011-08-1-0078, NSF grant CCF-0830320, and NSF CAREER grant CCF-1149225.}\\[3mm] \small{Department of Electrical Engineering and Computer Science, Colorado School of Mines}}

\date{June 2013; Revised April 2014}

\begin{document}

\maketitle

\vspace{-0.2in}

\begin{abstract}
Compressive Sensing (CS) exploits the surprising fact that the information contained in a sparse signal can be preserved in a small number of compressive, often random linear measurements of that signal. Strong theoretical guarantees have been established concerning the embedding of a sparse signal family under a random measurement operator and on the accuracy to which sparse signals can be recovered from noisy compressive measurements. In this paper, we address similar questions in the context of a different modeling framework. Instead of sparse models, we focus on the broad class of manifold models, which can arise in both parametric and non-parametric signal families. Using tools from the theory of empirical processes, we improve upon previous results concerning the embedding of low-dimensional manifolds under random measurement operators. We also establish both deterministic and probabilistic instance-optimal bounds in $\ell_2$ for manifold-based signal recovery and parameter estimation from noisy compressive measurements. In line with analogous results for sparsity-based CS, we conclude that much stronger bounds are possible in the probabilistic setting. Our work supports the growing evidence that manifold-based models can be used with high accuracy in compressive signal processing.
\end{abstract}


~


\section{Introduction}
\label{sec:intro}

\subsection{Concise signal models}

A significant byproduct of the Information Age has been an explosion
in the sheer quantity of raw data demanded from sensing systems.
From digital cameras to mobile devices, scientific computing to
medical imaging, and remote surveillance to signals intelligence,
the size (or dimension) $\dim$ of a typical desired signal continues
to increase. Naturally, the dimension $\dim$ imposes a direct burden
on the various stages of the data processing pipeline, from the data
acquisition itself to the subsequent transmission, storage, and/or
analysis.

Fortunately, in many cases, the information contained within a
high-dimensional signal actually obeys some sort of concise,
low-dimensional model. Such a signal may be described as having just
$\sparsity \ll \dim$ degrees of freedom for some $\sparsity$.
Periodic signals bandlimited to a certain frequency are one example;
they live along a fixed $\sparsity$-dimensional linear subspace of
$\real^\dim$. Piecewise smooth signals are an example of {\em sparse
signals}, which can be written as a succinct linear combination of
just $\sparsity$ elements from some basis such as a wavelet
dictionary. Still other signals may live along $\mdim$-dimensional
submanifolds of the ambient signal space $\real^\dim$; examples
include collections of signals observed from multiple viewpoints in
a camera or sensor network. In general, the conciseness of these
models suggests the possibility for efficient processing and
compression of these signals.

\subsection{Compressive measurements}

Recently, the conciseness of certain signal models has led to the
use of {\em compressive measurements} for simplifying the data
acquisition process. Rather than designing a sensor to measure a
signal $x \in \real^\dim$, for example, it often suffices to design
a sensor that can measure a much shorter vector $y = \proj x$, where
$\proj$ is a linear measurement operator represented as an $\pdim
\times \dim$ matrix, and where typically $\pdim \ll \dim$. As we
discuss below in the context of Compressive Sensing (CS), when
$\proj$ is properly designed, the requisite number of measurements
$\pdim$ typically scales with the information level $\sparsity$ of
the signal, rather than with its ambient dimension $\dim$.

\notechecked{edited this paragraph to clarify issues about storing/applying $\proj$} Surprisingly, the requirements on the measurement matrix $\proj$ can often be met by choosing $\proj$ randomly from an acceptable distribution. Most commonly, the entries of $\proj$ are chosen to be independent and identically distributed (i.i.d.)\ Gaussian random variables, although the use of structured random matrices is on the rise~\cite{krahmer2012suprema,eftekhari2012restricted}. Physical architectures have been proposed for hardware that will enable the acquisition of signals using compressive measurements~\cite{duarte2008spi,candes2008ics,healy2008cpi,demod}; many of these collect the compressive measurements $y$ of a signal $x$ directly, without explicitly computing a matrix multiplication on board. The potential benefits for data acquisition are numerous. These systems can enable simple, low-cost acquisition of a signal directly in compressed form without requiring knowledge of the signal structure in advance. Some of the many possible applications include distributed source coding in sensor networks~\cite{dcsJournalAccepted}, medical imaging~\cite{lustig2008csm}, and high-rate analog-to-digital conversion~\cite{candes2008ics,healy2008cpi,tropp2010beyond}. We note that, in all cases, the measurement matrix $\proj$ must be known to any decoder that will be used to process the compressed measurement vector $y$, but with suitable synchronization between the compressive measurement system and the decoder (e.g., exchanging a seed used to initialize a random number generator), it it not necessary for $\proj$ to be explicitly transmitted along with $y$.


\subsection{Signal understanding from compressive measurements}

Having acquired a signal $x$ in compressed form (in the form of a
measurement vector $y$), there are many questions that may then be
asked of the signal. These include: \notechecked{removed mention of sketching (formerly Q2), and replaced ``Q3'' with ``Q2'' for parameter estimation throughout the document}
\begin{itemize}
\item [Q1.] {\em Recovery:} What was the original signal $x$?
\item [Q2.] {\em Parameter estimation:} Supposing $x$ was generated from a
$\mdim$-dimensional parametric model, what was the original
$\mdim$-dimensional parameter that generated $x$?
\end{itemize}
Given only the measurements $y$ (possibly corrupted by noise),
solving either of the above problems requires exploiting the concise,
$\sparsity$-dimensional structure inherent in the
signal.\footnote{Other problems, such as finding the nearest
neighbor to $x$ in a large database of signals~\cite{JL_Indyk}, can
also be solved using compressive measurements and do not require
assumptions about the concise structure in $x$.} CS addresses
questions Q1 and Q2 under the assumption that the signal $x$ is
$\sparsity$-sparse (or approximately so) in some basis or
dictionary; in Section~\ref{sec:cs} we outline some key
theoretical bounds from CS regarding the accuracy to which these
questions may be answered.

\subsection{Manifold models for signal understanding}
\label{sec:mmodelsunder}

In this paper, we will address these questions in the context of a
different modeling framework for concise signal structure. Instead
of sparse models, we focus on the broad class of {\em manifold
models}, which arise both in settings where a $\mdim$-dimensional
parameter $\theta$ controls the generation of the signal and also in
non-parametric settings.

As a very simple illustration, consider the articulated signal in
Figure~\ref{fig:CSM}(a). We let $g(t)$ be a fixed continuous-time
Gaussian pulse centered at $t=0$ and consider a shifted version of
$g$ denoted as the parametric signal $f_\theta(t) := g(t-\theta)$
with $t,\theta \in [0,1]$. We then suppose the discrete-time signal
$x = x_\theta \in \real^\dim$ arises by sampling the continuous-time
signal $f_\theta(t)$ uniformly in time, i.e., $x_\theta(n) =
f_\theta(n/\dim)$ for $n=1,2,\dots,\dim$. As the parameter $\theta$
changes, the signals $x_\theta$ trace out a continuous
one-dimensional (1-D) curve $\manifold = \{x_\theta: \theta \in
[0,1]\} \subset \real^\dim$. The conciseness of our model (in
contrast with the potentially high dimension $\dim$ of the signal
space) is reflected in the low dimension of the path $\manifold$.

\begin{figure}[t]
\begin{center}
(a) \includegraphics[width=50mm]{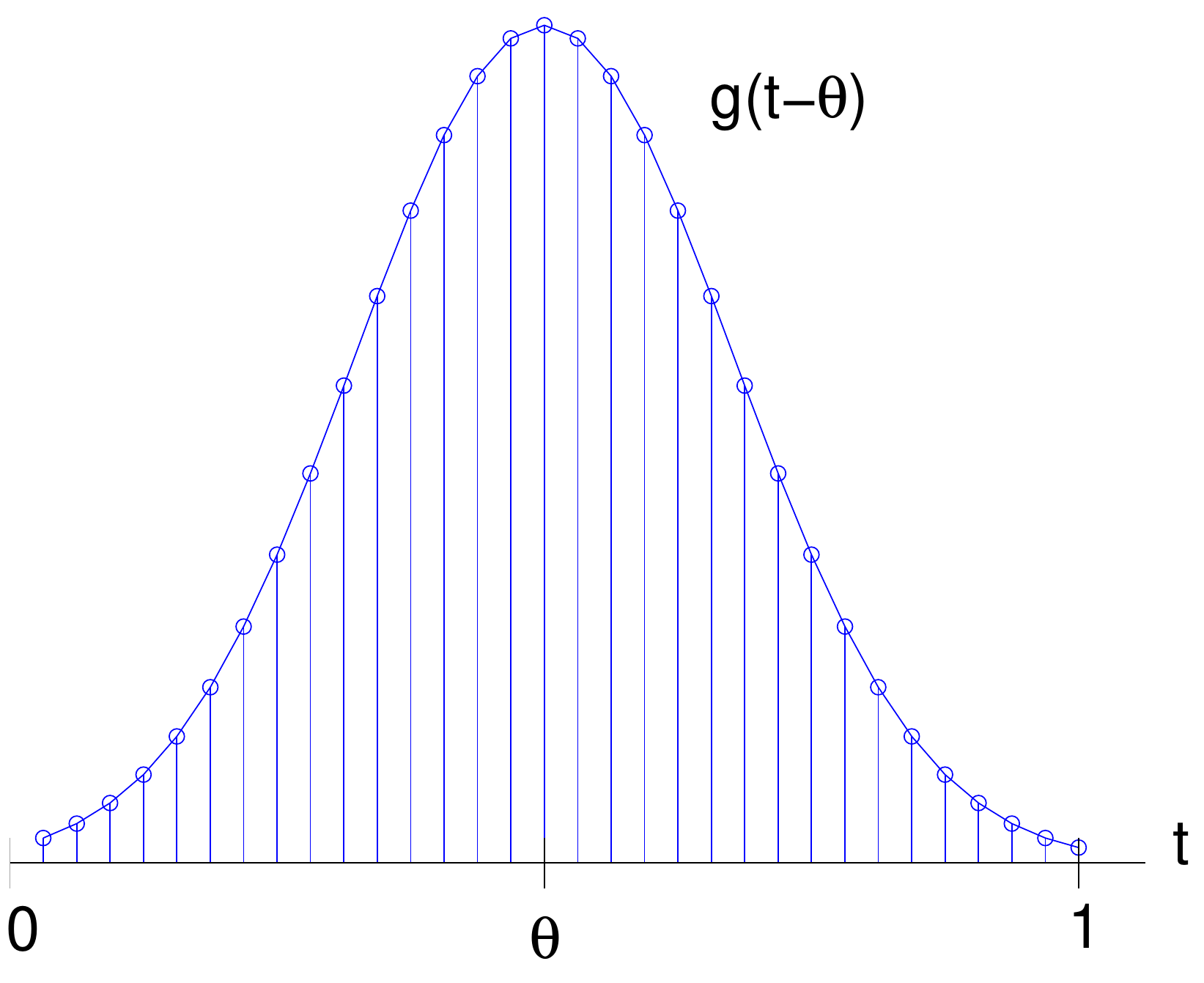} \quad\quad
(b) \includegraphics[width=60mm]{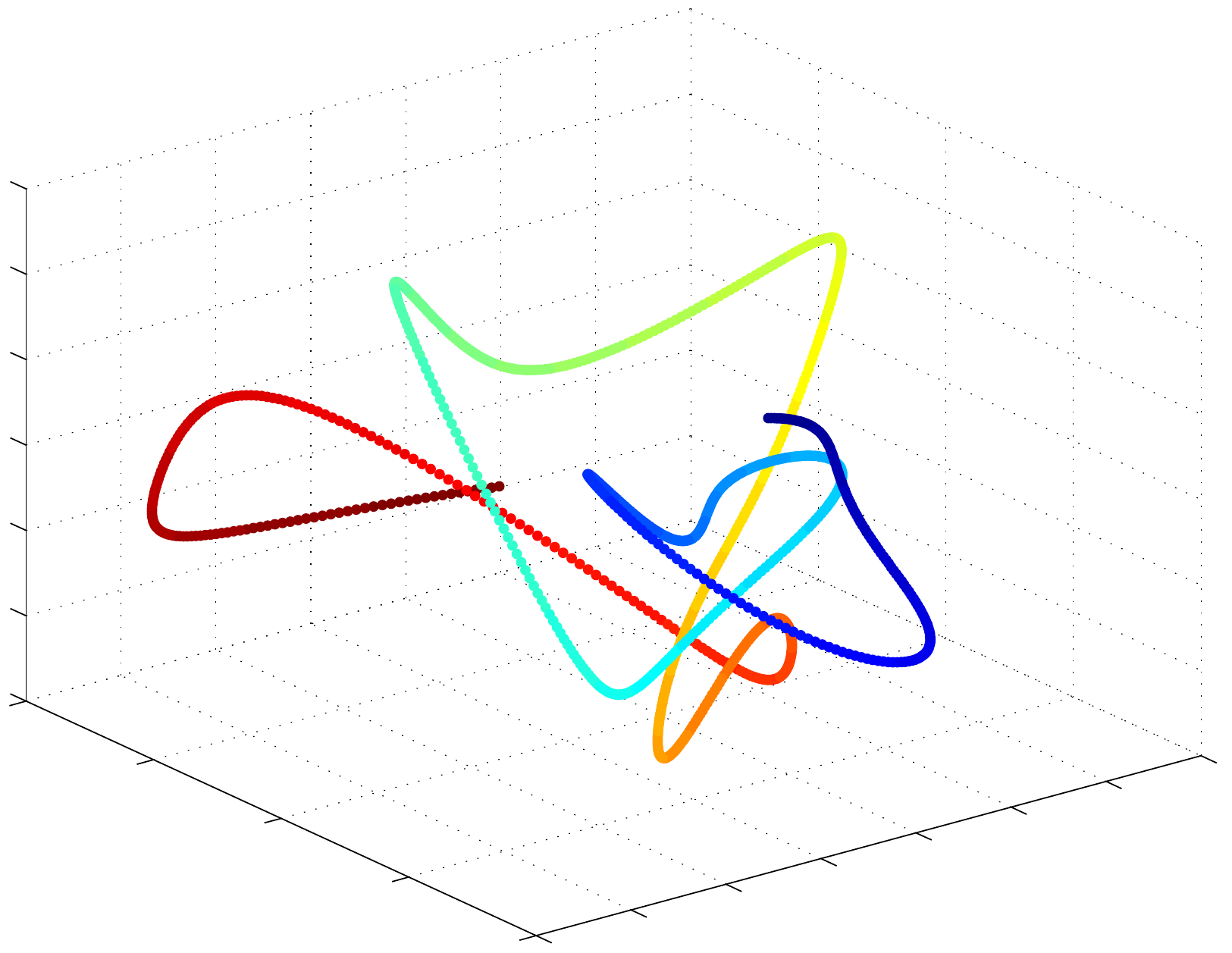}
\end{center}
\vspace*{-5mm} \caption{\small\sl \label{fig:CSM} (a) The
articulated signal $f_\theta(t) = g(t-\theta)$ is defined via shifts
of a primitive function $g$, where $g$ is a Gaussian pulse. Each
signal is sampled at $\dim$ points, and as $\theta$ changes, the
resulting signals trace out a 1-D manifold in $\real^\dim$. (b)
Projection of the manifold from $\real^\dim$ into $\real^3$ via a random $3 \times \dim$ matrix; the color/shading represents different values of $\theta
\in [0,1]$.}
\end{figure}

In the real world, manifold models may arise in a variety of
settings. A $\mdim$-dimensional parameter $\theta$ could reflect
uncertainty about the 1-D timing of the arrival of a signal (as in
Figure~\ref{fig:CSM}(a); see also~\cite{eftekhari2011matched}), the 2-D orientation and position of an
edge in an image, the 2-D translation of an image under study~\cite{park2012geometric}, the
multiple degrees of freedom in positioning a camera or sensor to
measure a scene~\cite{davenport2010joint}, the physical degrees of freedom in an articulated
robotic or sensing system, or combinations of the above. Manifolds
have also been proposed as approximate models for signal databases
such as collections of images of human faces or of handwritten
digits~\cite{Eigenfaces,digits,broomhead01wh}.

Consequently, the potential applications of manifold models are
numerous in signal processing. In some applications, the signal $x$
itself may be the object of interest, and the concise manifold model
may facilitate the acquisition or compression of that signal.
Alternatively, in parametric settings one may be interested in using
a signal $x = x_\theta$ to infer the parameter $\theta$ that
generated that signal.
In an application known as manifold learning, one may be presented
with a collection of data $\{x_{\theta_1}, x_{\theta_2}, \dots,
x_{\theta_n}\}$ sampled from a parametric manifold and wish to
discover the underlying parameterization that generated that
manifold.
Multiple manifolds can also be considered simultaneously, for
example in problems that require recognizing an object from one of
$n$ possible classes, where the viewpoint of the object is uncertain
during the image capture process. In this case, we may wish to know
which of $n$ manifolds is closest to the observed image $x$.

While any of these questions may be answered with full knowledge of the high-dimensional signal $x \in \real^\dim$, there is growing theoretical and experimental support that they can also be answered from only compressive measurements $y = \proj x$. In past work~\cite{baraniuk2009random}, we have shown that given a sufficient number $\pdim$ of random measurements, one can ensure with high probability that a manifold $\manifold \subset \real^\dim$ has a stable embedding in the measurement space $\real^\pdim$ under the operator $\proj$, such that pairwise Euclidean and geodesic distances are approximately preserved on its image $\proj \manifold$. We will discuss this in more detail later in Section~\ref{sec:mstable}, but a key aspect is that the number of requisite measurements $\pdim$ is linearly proportional to the information level of the signal, i.e., the dimension $\mdim$ of the manifold. In that work, the number of measurements was also logarithmically dependent on the ambient dimension $N$, although this dependence was later removed in the asymptotic case in~\cite{clarkson2008tighter} using a different set of assumptions on the manifold.

{\em The first contribution of this paper---presented in Section~\ref{sec:mstable}---is that we provide an improved lower bound on the number of random measurements to guarantee a stable embedding of a signal manifold.} In particular, we make the same assumptions on the manifold as in our past work~\cite{baraniuk2009random} but provide a measurement bound that is independent of the ambient dimension $N$.  Our bound is non-asymptotic, and we provide explicit constants. \notechecked{added next sentence.} Additionally we point out that this result is generic in the sense that it applies to any compact and smooth submanifold of $\mathbb{R}^{N}$ for which certain geometric properties (namely volume, dimension, and condition number) are known.

In order to do this, we use tools from the theory of empirical processes (namely, the idea of ``generic chaining''~\cite{talagrand2005generic}), which have recently been used to develop state-of-the-art RIP results for structured measurement matrices in CS~\cite{Rudelson-Dudley,tropp2010beyond,ra09,rombergGaussian2009,rombergNeelsh2010,rauhut10re,eftekhari2011matched,krahmer2012suprema}. More elementary arguments (e.g., involving simple concentration of measure inequalities) have previously been used in CS (see, e.g.,~\cite{JLCS}) for deriving RIP bounds for unstructured i.i.d.\ measurement matrices, and we also used such arguments in~\cite{baraniuk2009random} to derive a manifold embedding guarantee. However, it appears that the stronger machinery of the empirical process approach is necessary to derive stronger bounds, both in RIP problems and in manifold embedding problems. A chaining argument was employed in~\cite{clarkson2008tighter}, and in this paper we present a chaining argument that is suitable for studying the manifold embedding problem under our set of assumptions on the manifold. Because this chaining framework is fairly technical, we develop it entirely in the appendices so that the body of the paper will be as self-contained and expository as possible for someone seeking merely to understand the substance and context of our results. \notechecked{added next two sentences}(We do, however, include an example in the body of the paper to provide insight into the machinery developed in the appendices.) We also observe that similar results are attainable through the use of the Dudley inequality~\cite{ledoux2011probability}, but a direct argument (as the one presented here) has the advantages of potentially better exploiting the geometry of the model and therefore producing tighter bounds, offering improved insight into the problem, and being more amenable to future improvements to our arguments.


As a very simple illustration of the embedding phenomenon, Figure~\ref{fig:CSM}(b) presents an experiment where just $\pdim = 3$ compressive measurements are acquired from each point $x_\theta$ described in Figure~\ref{fig:CSM}(a). We let $\dim = 1024$ and construct a randomly generated $3 \times \dim$ matrix $\proj$ whose entries are i.i.d.\ Gaussian random variables with zero mean and variance of $1/3$. Each point $x_\theta$ from the original manifold $\manifold \subset \real^{1024}$ maps to a unique point $\proj x_\theta$ in $\real^3$; the manifold embeds in the low-dimensional measurement space. Given any $y = \proj x_{\theta'}$ for $\theta'$ unknown, then, it is possible to infer the value $\theta'$ using only knowledge of the parametric model for $\manifold$ and the
measurement operator $\proj$. Moreover, as the number $\pdim$ of compressive measurements increases, the manifold embedding becomes much more stable and remains highly self-avoiding.

Indeed, there is strong evidence that, as a consequence of this phenomenon, questions such as Q1 (signal recovery) and Q2 (parameter estimation) can be accurately solved using only compressive measurements of a signal $x$, and that these procedures are robust to noise and to deviations of the signal $x$ away from the manifold $\manifold$~\cite{mbwPhdThesis,davenport2007sfc,shah2011iterative}. Additional theoretical and empirical justification has followed for
the manifold learning~\cite{hegde2007rpm} and multiclass recognition problems~\cite{davenport2007sfc} described above. Consequently, many of the advantages of  compressive measurements that are beneficial in sparsity-based CS (low-cost sensor design, reduced transmission requirements, reduced storage requirements,
lack of need for advance knowledge of signal structure, simplified computation in the low-dimensional space $\real^\pdim$, etc.)\ may also be enjoyed in settings where manifold models capture the concise signal structure. Moreover, the use of a manifold model can often capture the structure of a signal in many fewer degrees of freedom $\mdim$ than would be required in any sparse representation, and thus the measurement rate $\pdim$ can be greatly reduced
compared to sparsity-based CS approaches.

\notechecked{updated next sentence} {\em The second contribution of this paper---presented in Section~\ref{sec:srpe}---is that we establish theoretical bounds on the accuracy to which questions Q1 (signal recovery) and Q2 (parameter estimation) may be answered.} To do this, we rely largely on the new analytical chaining framework described above. We consider both deterministic and probabilistic instance-optimal bounds, and we see strong similarities to analogous results that have been derived for sparsity-based CS. As with sparsity-based CS, we show for manifold-based CS that for any fixed $\proj$, uniform deterministic $\ell_2$ recovery bounds for recovery of all $x$ are necessarily poor. We then show that, as with sparsity-based CS, providing for any $x$ a probabilistic bound that holds over most $\proj$ is possible with the desired accuracy. We consider both noise-free and noisy measurement settings and compare
our bounds with sparsity-based CS. \notechecked{added next sentence.} Finally, it should be noted that our results concerning question Q1 are independent of the parametrization of the manifold, whereas, in contrast, our results concerning question Q2 are specific to the given parametrization of the manifold.

{\em We feel that a third contribution of this paper comes in the form of the analytical tools we use to study the above problems.} Our chaining argument allows us to study not only the embedding problem (as in~\cite{clarkson2008tighter}) but also Q1 and Q2. Moreover, in Appendix~\ref{sec:useful diff geom reuslts}, which we call the ``Toolbox,'' we present a collection of implications of our assumption that the manifold has bounded {\em condition number} (see Section~\ref{sec:mmodels} for definition). This elementary property, also known as the {\em reach} of a manifold in the geometric measure theory literature~\cite{federer1959curvature}, \notechecked{updated cite for reach} has become somewhat popular in the analysis of manifold models for signal processing (e.g., see~\cite{niyogi,baraniuk2009random,yap2013stable,iwen2012approximation,verma2012,davenport2010joint,davenport2007sfc}). The seminal paper~\cite{niyogi} (also see \cite{federer1959curvature}) \notechecked{added cite} contains a collection of implications of bounded condition number that have been used directly or indirectly in numerous works, including~\cite{baraniuk2009random,yap2013stable,iwen2012approximation,verma2012,davenport2010joint,davenport2007sfc}. We restate some of these implications in the Toolbox. Unfortunately, after very careful study we were unable to confirm for ourselves some of the original proofs appearing in~\cite{niyogi}. Therefore, some of the statements and proofs in the Toolbox below differ slightly from their original counterparts in~\cite{niyogi}. We hope that these results will provide a useful reference for the continued study of manifolds with bounded condition number.

\subsection{Paper organization}

Section~\ref{sec:background} provides the necessary background on sparsity-based CS and on manifold models to place our work in the proper context. In Section~\ref{sec:mstable}, we state our improved bound regarding stable embeddings of manifolds. In Section~\ref{sec:srpe}, we then formalize our criteria for answering questions Q1 and Q2 in the context of manifold models. We first confront the task of deriving deterministic instance-optimal bounds in $\ell_2$ and then consider probabilistic instance-optimal bounds in $\ell_2$. We conclude in Section~\ref{sec:concl} with a final discussion. The Toolbox (Appendix~\ref{sec:useful diff geom reuslts}) establishes a collection of useful results in differential geometry that are frequently used throughout our technical proofs, which appear in the remaining appendices.

\section{Background}
\label{sec:background}

\subsection{Sparsity-Based Compressive Sensing} \label{sec:cs}

\subsubsection{Sparse models}

The concise modeling framework used in CS is
{\em sparsity}. Consider a signal $x \in \real^\dim$ and suppose the
$\dim \times \dim$ matrix $\Psi = [\psi_1 ~ \psi_2 ~ \cdots~
\psi_N]$ forms an orthonormal basis for $\real^\dim$. We say $x$ is
$\sparsity$-sparse in the basis $\Psi$ if for $\alpha \in
\real^\dim$ we can write $x = \Psi \alpha$, where $\|\alpha\|_0 = \sparsity < \dim$. (The $\ell_0$-norm notation
counts the number of nonzeros of the entries of $\alpha$.) In a
sparse representation, the actual information content of a signal is
contained exclusively in the $\sparsity < \dim$ positions and values
of its nonzero coefficients.

For those signals that are approximately sparse, we may measure
their proximity to sparse signals as follows. We define
$\alpha_\sparsity \in \real^\dim$ to be the vector containing only
the largest $\sparsity$ entries of $\alpha$ in magnitude, with the remaining
entries set to zero. Similarly, we let $x_\sparsity = \Psi
\alpha_\sparsity$. It is then common to measure the proximity to
sparseness using either $\|\alpha-\alpha_\sparsity\|_1$ or
$\norm{\alpha-\alpha_\sparsity}$ (the latter of which equals
$\norm{x-x_\sparsity}$ because $\Psi$ is orthonormal). Here and elsewhere in this paper, $\|\cdot\|$ stands for the $\ell_2$ norm.

\subsubsection{Stable embeddings of sparse signal families}
\label{sec:csmeas}

CS uses the concept of sparsity to simplify the data acquisition
process. Rather than designing a sensor to measure a signal $x \in
\real^\dim$, for example, it often suffices to design a sensor that
can measure a much shorter vector $y = \proj x$, where $\proj$ is a
linear measurement operator represented as an $\pdim \times \dim$
matrix, and typically $\pdim \ll \dim$.

The measurement matrix $\proj$ must have certain properties in order
to be suitable for CS. One desirable property (which leads to the
theoretical results we mention in Section~\ref{sec:csresults}) is
known as the Restricted Isometry Property
(RIP)~\cite{CandesUES,CandesECLP,CandesSSR}. We say a matrix $\proj$
meets the {\em RIP of order $\sparsity$ with respect to the basis
$\Psi$} if for some $\Uresolution_\sparsity > 0$,
$$
(1-\Uresolution_\sparsity) \norm{\alpha} \le \norm{\proj \Psi \alpha} \le
(1+\Uresolution_\sparsity) \norm{\alpha}
$$
holds for all $\alpha \in \real^\dim$ with $\|\alpha\|_0 \le
\sparsity$. Intuitively, the RIP can be viewed as guaranteeing a
{\em stable embedding} of the collection of $\sparsity$-sparse
signals within the measurement space $\real^\pdim$. In particular,
supposing the RIP of order $2\sparsity$ is satisfied with respect to
the basis $\Psi$, then for all pairs of $\sparsity$-sparse signals
$x_1, x_2 \in \real^\dim$, we have
\begin{equation}
(1-\Uresolution_{2\sparsity}) \norm{x_1-x_2} \le \norm{\proj x_1 - \proj
x_2} \le (1+\Uresolution_{2\sparsity}) \norm{x_1-x_2}. \label{eq:rip2}
\end{equation}

\notechecked{edited the next sentence.} Although deterministic constructions of matrices meeting the RIP with few rows (ideally proportional to the sparsity level $K$) are still a work in progress, it is known that the RIP can often be met by
choosing $\proj$ randomly from an acceptable distribution. For
example, let $\Psi$ be a fixed orthonormal basis for $\real^\dim$
and suppose that
\begin{equation}
 \pdim \ge \Cl{RIPConst} \sparsity
\log(\dim/\sparsity) \label{eq:nummeas}
\end{equation}
for some constant $\Cr{RIPConst}$. Then supposing that the entries of
the $\pdim \times \dim$ matrix $\proj$ are drawn as i.i.d.\ Gaussian random variables with mean $0$ and
variance $\frac{1}{\pdim}$, it follows that with high probability
$\proj$ meets the RIP of order $\sparsity$ with respect to the basis
$\Psi$. Two aspects of this construction deserve special notice:
first, the number $\pdim$ of measurements required is linearly
proportional to the information level $\sparsity$ (and logarithmic in the ambient dimension $N$), and second,
neither the sparse basis $\Psi$ nor the locations of the nonzero
entries of $\alpha$ need be known when designing the measurement
operator $\proj$. Other random distributions for $\proj$ may also be
used, all requiring approximately the same number of measurements~\cite{rauhut2010compressive,krahmer2012suprema,eftekhari2012restricted}.

\subsubsection{Sparsity-based signal recovery}
\label{sec:csresults}

\notechecked{removed ``and sketching'' from section name} Although the sparse structure of a signal $x$ need not be known when
collecting measurements $y = \proj x$, a hallmark of CS is the use
of the sparse model in order to facilitate understanding from the
compressive measurements. A variety of algorithms have been proposed
to answer Q1 (signal recovery), where we seek to solve the
apparently undercomplete set of $\pdim$ linear equations $y = \proj
x$ for $\dim$ unknowns. The canonical
method~\cite{DonohoCS,CandesUES,CandesRUP} is known as {\em
$\ell_1$-minimization} and is formulated as follows: first solve
\begin{equation}
\widehat{\alpha} = \argmin_{\alpha' \in \real^\dim} \|\alpha'\|_1
~\mathrm{subject~to}~ y = \proj \Psi \alpha', \label{eq:l1min}
\end{equation}
and then set $\widehat{x} = \Psi \widehat{\alpha}$. This recovery program can also be extended to account for measurement noise. The following bound is known.
\begin{thm}{\em \cite{candes2008rip}} Suppose that $\proj$
satisfies the RIP of order $2\sparsity$ with respect to $\Psi$ and
with constant $\Uresolution_{2\sparsity} < \sqrt{2}-1$. Let $x \in
\real^\dim$, and suppose that $y = \proj x + \noise$ where $\norm{\noise} \le \epsilon$. Then let
$$ \widehat{\alpha} =
\arg\min_{\alpha' \in \real^\dim} \|\alpha'\|_1
~\mathrm{subject~to}~ \norm{y - \proj \Psi \alpha'} \le \epsilon,
$$
and set $\widehat{x} = \Psi \widehat{\alpha}$. Then
\begin{equation}
\norm{x-\widehat{x}} = \norm{\alpha-\widehat{\alpha}} \le \CandesOne
\sparsity^{-\frac{1}{2}} \|\alpha-\alpha_\sparsity\|_1 + \CandesTwo
\epsilon. \label{eq:mixed2}
\end{equation}
for constants $\CandesOne$ and $\CandesTwo$.
 \label{theo:mixed2}
\end{thm}
This result is not unique to $\ell_1$ minimization; similar
bounds have been established for signal recovery using greedy
iterative algorithms OMP \cite{davenport2010analysis}, ROMP~\cite{romp}, and CoSAMP~\cite{cosamp}.
Bounds of this type are extremely encouraging for signal processing.
From only $\pdim$ measurements, it is possible to recover $x$ with
quality that is comparable to its proximity to the nearest
$\sparsity$-sparse signal, and if $x$ itself is $\sparsity$-sparse
and there is no measurement noise, then $x$ can be recovered
exactly. Moreover, despite the apparent ill-conditioning of the
inverse problem, the measurement noise is not dramatically amplified
in the recovery process.

These bounds are known as {\em deterministic}, {\em
instance-optimal} bounds because they hold deterministically for any
$\proj$ that meets the RIP, and because for a given $\proj$ they
give a guarantee for recovery of any $x \in \real^\dim$ based on its
proximity to the concise model.

The use of $\ell_1$ as a measure for proximity to the concise model
(on the right hand side of (\ref{eq:mixed2}))
arises due to the difficulty in establishing $\ell_2$ bounds on the
right hand side. Indeed, it is known that deterministic $\ell_2$
instance-optimal bounds cannot exist that are comparable to (\ref{eq:mixed2}). In particular, for any
$\proj$, to ensure that $\norm{x-\widehat{x}} \le \CddOne
\norm{x-x_\sparsity}$ for all $x$, it is known~\cite{CDD} that this
requires that $\pdim \ge \CddTwo \dim$ regardless of $\sparsity$.

However, it is possible to obtain an instance-optimal $\ell_2$ bound
for sparse signal recovery in the noise-free setting by changing
from a deterministic formulation to a {\em probabilistic}
one~\cite{CDD,devore08in}. In particular, by considering any given
$x \in \real^\dim$, it is possible to show that for {\em most}
random $\proj$, letting the measurements $y = \proj x$, and
recovering $\widehat{x}$ via $\ell_1$-minimization (\ref{eq:l1min}),
it holds that
\begin{equation}
\norm{x-\widehat{x}} \le \CddThree
\norm{x-x_\sparsity}.\label{eq:devore}
\end{equation}
While the proof of this statement~\cite{devore08in} does not involve
the RIP directly, it holds for many of the same random distributions
that work for RIP matrices, and it requires the same number of
measurements (\ref{eq:nummeas}) up to a constant.

\notechecked{removed old ``Q2'' notation from mention of sketching} Similar bounds hold for the closely related problem of ``sketching,'' where the goal is to use the compressive measurement vector $y$ to identify and report only approximately $\sparsity$ expansion coefficients that best describe the original signal, i.e., a sparse approximation to $\alpha_\sparsity$. In the case where $\Psi = I$, an efficient randomized measurement process coupled with a customized recovery algorithm~\cite{gilbert2007osa} provides signal sketches that meet a deterministic mixed-norm $\ell_2/\ell_1$ instance-optimal bound analogous to~(\ref{eq:mixed2}) in the noise-free setting. A desirable aspect of this construction is that the computational complexity scales with only $\log(\dim)$ (and is polynomial in $\sparsity$); this is possible because only approximately $\sparsity$ pieces of information must be computed to describe the signal. Though at a higher computational cost, the aforementioned greedy algorithms (such as CoSaMP) for signal recovery can also be interpreted as sketching techniques in that they produce explicit sparse approximations to $\alpha_\sparsity$.
Finally, for signals that are sparse in the Fourier domain ($\Psi$ consists of the DFT vectors), probabilistic $\ell_2/\ell_2$ instance-optimal bounds have
been established for a specialized sketching algorithm~\cite{gilbert05im,gilbert2008tutorial} that are analogous to (\ref{eq:devore}).

\subsection{Manifold models and properties}
\label{sec:mmodels}

\subsubsection{Overview}

\notechecked{created new subsubsection heading}

As we have discussed in Section~\ref{sec:mmodelsunder}, there are
many possible modeling frameworks for capturing concise signal
structure. Among these possibilities are the broad class of manifold
models.

Manifold models arise, for example, in settings where the signals of
interest vary continuously as a function of some $\mdim$-dimensional
parameter. Suppose, for instance, that there exists some parameter
$\theta$ that controls the generation of the signal. We let
$x_\theta \in \real^\dim$ denote the signal corresponding to the
parameter $\theta$, and we let $\Theta$ denote the
$\mdim$-dimensional parameter space from which $\theta$ is drawn. In
general, $\Theta$ itself may be a $\mdim$-dimensional manifold and
need not be embedded in an ambient Euclidean space. For example,
supposing $\theta$ describes the 1-D rotation parameter in a
top-down satellite image, we have $\Theta = \mathbb{S}^1$.

Under certain conditions on the parameterization $\theta \mapsto
x_\theta$, it follows that $\manifold := \{x_\theta: \theta \in \Theta\}$
forms a $\mdim$-dimensional {\em submanifold} of $\real^\dim$. An
appropriate visualization is that the set $\manifold$ forms a
nonlinear $\mdim$-dimensional ``surface'' within the
high-dimensional ambient signal space $\real^\dim$. Depending on the
circumstances, we may measure the distance between points two points
$x_{\theta_1}$ and $x_{\theta_2}$ on the manifold $\manifold$ using
either the ambient Euclidean distance $\norm{x_{\theta_1}-x_{\theta_2}}$
or the geodesic distance along the manifold, which we denote as
$\gdist(x_{\theta_1},x_{\theta_2})$. In the case where the geodesic
distance along $\manifold$ equals the native distance in parameter
space, i.e., when
\begin{equation}
\gdist(x_{\theta_1},x_{\theta_2}) = d_\Theta(\theta_1, \theta_2),
\label{eq:isometry}
\end{equation}
we say that $\manifold$ is {\em isometric} to $\Theta$. The
definition of the distance $d_\Theta(\theta_1, \theta_2)$ depends on
the appropriate metric for the parameter space $\Theta$; supposing
$\Theta$ is a convex subset of Euclidean space, then we can let
$d_\Theta(\theta_1, \theta_2) = \norm{\theta_1-\theta_2}$.

While our discussion above concentrates on the case of manifolds
$\manifold$ generated by underlying parameterizations, we stress
that manifolds have also been proposed as approximate
low-dimensional models within $\real^\dim$ for nonparametric signal
classes such as images of human faces or handwritten
digits~\cite{Eigenfaces,digits,broomhead01wh}. These signal families
may also be considered.

The results we present in this paper will make reference to certain
characteristic properties of the manifold under study. These terms
are originally defined in~\cite{niyogi,baraniuk2009random} and are repeated
here for completeness. First, our results will depend on a measure
of regularity for the manifold. For this purpose, we adopt the notion of the {\em
condition number} of a manifold, which is also known as the reach of a manifold in the geometric measure theory literature~\cite{federer1959curvature,niyogi}.

\begin{definition}{\em \cite{niyogi}}
Let $\manifold$ be a compact Riemannian submanifold of $\real^\dim$.
The {\em condition number} is defined as $1/\condition$, where
$\condition$ is the largest number having the following property:
The open normal bundle about $\manifold$ of radius $r$ is embedded
in $\real^\dim$ for all $r < \condition$. \label{def:cn}
\end{definition}

The condition number $1/\condition$ controls both local properties
and global properties of the manifold. Its role is summarized in two
key relationships (see the Toolbox and~\cite{niyogi} for more detail). First, the the curvature of any
unit-speed geodesic path on $\manifold$ is bounded by
$1/\condition$. Second, at long geodesic distances, the condition
number controls how close the manifold may curve back upon itself.
For example, supposing $x_1,x_2 \in \manifold$ with $\gdist(x_1,x_2)
> \condition$, it must hold that $\norm{x_1-x_2} > \condition/2$.

We continue with a brief but concrete example to illustrate specific
values for these quantities. Let $\dim
> 0$, $\kappa > 0$, $\Theta = \real \!\! \mod 2\pi$, and suppose $x_\theta
\in \real^\dim$ is given by
$$
x_\theta = [\kappa \cos(\theta); ~ \kappa \sin(\theta); ~ 0; ~0;
\cdots 0]^T.
$$
In this case, $\manifold = \{x_\theta: \theta \in \Theta\}$ forms a
circle of radius $\kappa$ in the $x(1), x(2)$ plane. The manifold
dimension $\mdim=1$, and the condition number $1/\condition = 1/\kappa$. 
We also refer in our results to the $\mdim$-dimensional volume of \notechecked{removed ``the''} $\manifold$, denoted by $\volume$, which
in this example corresponds to the circumference $2\pi \kappa$ of
the circle.

\notechecked{added the following example and subsubsection.} We conclude this section with a less trivial example of computing the condition number (or, alternatively, reach).

\subsubsection{Example: Complex exponential curve}
\label{sec:complexexpcurve}

For an integer $f_{C}$, set $N:=2f_{C}+1$. Let $\beta:\mathbb{R}\rightarrow\mathbb{C}^{N}$ denote the complex exponential curve defined as
\begin{equation}\label{eq:def of sine manifold}
\beta_{t}=\beta(t)=\left[\begin{array}{c}
e^{-\uniti2\pi f_{C}t}\\
e^{-\uniti2\pi(f_{C}-1)t}\\
\vdots\\
e^{\uniti2\pi(f_{C}-1)t}\\
e^{\uniti2\pi f_{C}t}
\end{array}\right]
\end{equation}
for $t\in\mathbb{R}$. The following result, proved in Appendix \ref{sec:example}, gives an estimate of the condition number (which we denote here by  $1/\tau_{\beta}$) of the complex exponential curve $\beta$.\footnote{Unlike $\beta$, which is a subset of $\mathbb{C}^N$, its real or imaginary parts live in $\mathbb{R}^N$ and are perhaps more consistent with the rest of this paper (which studies submanifolds of $\mathbb{R}^N$). However, finding the condition number of $\mbox{re}(\beta)$ or $\mbox{im}(\beta)$ is far more tedious and therefore not pursued here for the sake of the clarity of our exposition.} The reader may refer to~\cite{yap2013geometric} for related computations concerning the complex exponential curve.

\begin{lemma}\label{lemma:example}
For the complex exponential curve $\beta$ in $\mathbb{C}^N$ (as defined in \eqref{eq:def of sine manifold}), let $1/\tau_\beta$ denote its condition number. Then, for some integer $N_{\text{sine}}$ and (known) constant $\alpha_{\text{sine}}<1$, the following holds if $N>N_{\text{sine}}$:
$$
\alpha_{\text{sine}}\sqrt{N} \le \tau_{\beta} \le \sqrt{N}.
$$
\end{lemma}

\section{Stable embeddings of manifolds}
\label{sec:mstable}

In cases where the signal class of interest $\manifold$ forms a low-dimensional submanifold of $\real^\dim$, we have theoretical justification that the information necessary to distinguish and recover signals $x \in \manifold$ can be well-preserved under a sufficient number of compressive measurements $y = \proj x$. In particular, it was first shown in~\cite{baraniuk2009random} that an RIP-like property holds for families of manifold-modeled signals. The result stated that, under a random projection onto $\real^M$, pairwise distances between the points on $\manifold$ are approximately preserved with high probability, provided that $M$, \notechecked{added ``the''} the number of measurements, is large enough. Mainly, $M$ should scale linearly with the dimension $K$ of $\manifold$ and logarithmically with the ambient dimension $N$. The dependence on $N$ was later removed in~\cite{clarkson2008tighter}, which used a different set of assumptions on the manifold to help derive a sharper lower bound on the requisite number of random measurements. Unfortunately, the results given in \cite{clarkson2008tighter} hold only as the isometry constant $\epsilon \rightarrow 0$ in \eqref{eq:rip2}, with asymptotic threshold and constants fixed but unspecified. The manifold properties assumed in~\cite{clarkson2008tighter} are arguably more complicated and less commonly used than the manifold volume and condition number which are at the heart of our results. (On the other hand, there do exist manifolds where the properties assumed in~\cite{clarkson2008tighter} clearly provide a stronger analysis.)

In this section, we establish an improved lower bound on $M$ to ensure a stable embedding of a manifold. We make the same assumptions on the manifold as in our past work~\cite{baraniuk2009random} but provide a measurement bound that is independent of the ambient dimension $N$.  Our bound holds for every $\epsilon\le1/3$ and we provide explicit constants. The proof, presented in Appendix~\ref{sec:Proof of new manifold embedding}, draws from the ideas in generic
chaining \cite{talagrand2005generic}, which have been recently used to develop state-of-the-art RIP results for structured measurement matrices in CS~\cite{Rudelson-Dudley,tropp2010beyond,ra09,rombergGaussian2009,rombergNeelsh2010,rauhut10re,eftekhari2011matched,krahmer2012suprema}. As in~\cite{clarkson2008tighter}, we control the failure probability of the manifold embedding by forming a so-called {\em chain} on a sequence of increasingly finer covers on the index set of the random process~\cite{talagrand2005generic,ledoux2011probability}. Aside from delivering an improved bound (and also allowing us to study Q1 and Q2 in Section~\ref{sec:srpe}), we hope that our exposition in this paper will encourage yet more researchers in the field of CS to use this powerful technique.

\begin{thm} \label{thm:new manifold embedding}
	Let $\manifold$ be a compact $\mdim$-dimensional Riemannian
	submanifold of $\real^\dim$ having condition number $1/\condition$ and
	volume $\volume$. Conveniently assume that\footnote{\notechecked{Edited the footnote} Theorem~\ref{thm:new manifold embedding} still holds, with a worse (larger) lower bound in \eqref{eq:mmeasmain}, after relaxing the assumption in \eqref{eq:mild assump on volume 1 main}. One example of a manifold that does satisfy the assumption in \eqref{eq:mild assump on volume 1 main} is the complex exponential curve described in Section~\ref{sec:complexexpcurve}, as long as $N \ge 7$.}
	\begin{equation}\label{eq:mild assump on volume 1 main}
		\frac{\volume}{\condition^K}\ge \left(\frac{21}{2\sqrt{K}}\right)^{K}.
	\end{equation}
	Fix $0 <\eps \le 1/3$ and $0 < \rho < 1$. Let $\proj$ be a random $\pdim \times
	\dim$ matrix populated with i.i.d.\ zero-mean Gaussian random variables with variance of $1/M$ with
	\begin{equation}
		M  \ge 18\epsilon^{-2}\max\left(24K+2K\log\left(\frac{\sqrt{K}}{\condition\epsilon^{2}}\right)+\log(2\volume^{2}),\log\left(\frac{8}{\failprob}\right)\right).
		\label{eq:mmeasmain}
	\end{equation}
	Then with probability at least $1-\rho$ the
	following statement holds: For every pair of points $x_1,x_2 \in
	\manifold$,
	\begin{equation}
		(1-\epsilon) \norm{x_1-x_2} \le \norm{\proj x_1 - \proj x_2} \le
		(1+\epsilon) \norm{x_1-x_2}. \label{eq:mmain}
	\end{equation}
\label{theo:manifoldjl}
\end{thm}
The proof of the above result can be found in Appendix~\ref{sec:Proof of new manifold embedding}. In essence, manifolds with higher volume or
with greater curvature have more complexity, which leads to an increased number of measurements~(\ref{eq:mmeasmain}).
By comparing (\ref{eq:rip2}) with (\ref{eq:mmain}), we see a strong
analogy to the RIP of order $2\sparsity$. This theorem establishes
that, like the class of $\sparsity$-sparse signals, a collection of
signals described by a $\mdim$-dimensional manifold $\manifold
\subset \real^\dim$ can have a stable embedding in an
$\pdim$-dimensional measurement space. Moreover, the requisite
number of random measurements $\pdim$ is once again almost linearly
proportional to the information level (or number of degrees of
freedom) $\mdim$. It is important to note that in \eqref{eq:mmeasmain}, the combined dependence on the manifold dimension $K$, condition number $1/\condition$, and	volume $\volume$ cannot, generally speaking, be improved. \notechecked{edited the rest of this paragraph for clarity} In particular, consider the case where $\mathcal{M}$ is a $K$-dimensional unit ball in $\mathbb{R}^{N}$, that is $\mathcal{M}=\mathcal{B}_{K}$. Clearly, in this case, $\tau=1$. Additionally, from \eqref{eq:bnds on VkBk}, we observe that $V_{\mathcal{M}}=V_{\mathcal{B}_{K}}\propto K^{-K/2}$
and so $\log V_{\mathcal{M}}^{2}\propto -K\log K$. As a result, plugging for $V_{\mathcal{M}}$ back into \eqref{eq:mmain} cancels the term $2K\log\sqrt{K}=K\log K$ on the right hand side of \eqref{eq:mmeasmain}. It follows that the lower bound in \eqref{eq:mmeasmain} scales with $K$ (rather than $K\log K$) in this case. We conclude that, in this special case, the lower bound in \eqref{eq:mmain} is optimal (up to a constant factor).

As was the case with the RIP for sparse signal processing, this sort of result has a number of possible implications for manifold-based signal processing. First, individual signals obeying a manifold model can be acquired and stored efficiently using compressive measurements, and it is unnecessary to employ the manifold model itself as part of the compression process. Rather, the model needs only to be used for signal understanding from the compressive measurements. Second, problems such as Q1 (signal recovery) and Q2 (parameter estimation) can be addressed, as we discuss in Section~\ref{sec:srpe}. Aside from this theoretical analysis, we have reported promising experimental recovery/estimation results with various classes of parametric signals~\cite{mbwPhdThesis,davenport2007sfc}. Also, taking a different analytical perspective (a statistical one, assuming additive white Gaussian measurement noise), estimation-theoretic quantities such as the Cramer-Rao Lower Bound (for a specialized set of parametric problems) have been shown to be preserved in the compressive measurement space as a consequence of the stable embedding~\cite{ramasamy2012compressive}. Third, the stable embedding results can also be extended to the case of multiple manifolds that are simultaneously embedded~\cite{davenport2007sfc}; this allows both the classification of an observed object to one of several possible models (different manifolds) and the estimation of a parameter within that class (position on a manifold). Fourth, collections of signals obeying a manifold model (such as multiple images of a scene photographed from different perspectives) can be acquired using compressive measurements, and the resulting manifold structure will be preserved among the suite of measurement vectors in $\real^\pdim$ \cite{davenport2010joint,park2012geometric}. Fifth, we have provided empirical and theoretical support for the use of manifold learning in the reduced-dimensional space~\cite{hegde2007rpm}; this can dramatically simplify the computational and storage demands on a system for processing large databases of signals.

\notechecked{in the remainder of this section, added example to give reader some insight and to offer pointers into the various proof sections} Before presenting an application of Theorem~\ref{thm:new manifold embedding} in the next section, we would like to outline the chaining argument used in its proof through an example. Suppose that $\manifold$ is the unit circle embedded in $\mathbb{R}^N$ and that we observe this manifold through a measurement operator $\Phi\in\mathbb{R}^{M\times N}$. To study the quality of embedding, we first need to identify the set of all secants connecting two points in $\manifold$. In this example, the set of all normalized secants of $\manifold$ (which we denote by $U(\manifold)$) also forms a unit circle and equals $\manifold$, i.e., $U(\manifold)=\manifold$. Let $\{C_j\}$ be a sequence of increasingly finer covers on $\manifold$ (or equivalently on $U(\manifold)$). Constructing a sequence of covers for the secants of a manifold in general is studied in Appendix~\ref{sec:Constructing a Sequence of Covers for U(M)}.  For an arbitrary normalized secant $y=(x_1-x_2)/\|x_1-x_2\|$ with $x_1,x_2\in\manifold$, let $\pi_j(y)$ represent the nearest point to $y$ on the $j$th cover (for every $j\ge 0$). This construction is illustrated in Figure~\ref{fig:Example}.
\begin{figure}[t]
\begin{center}
\includegraphics[width=5.5in]{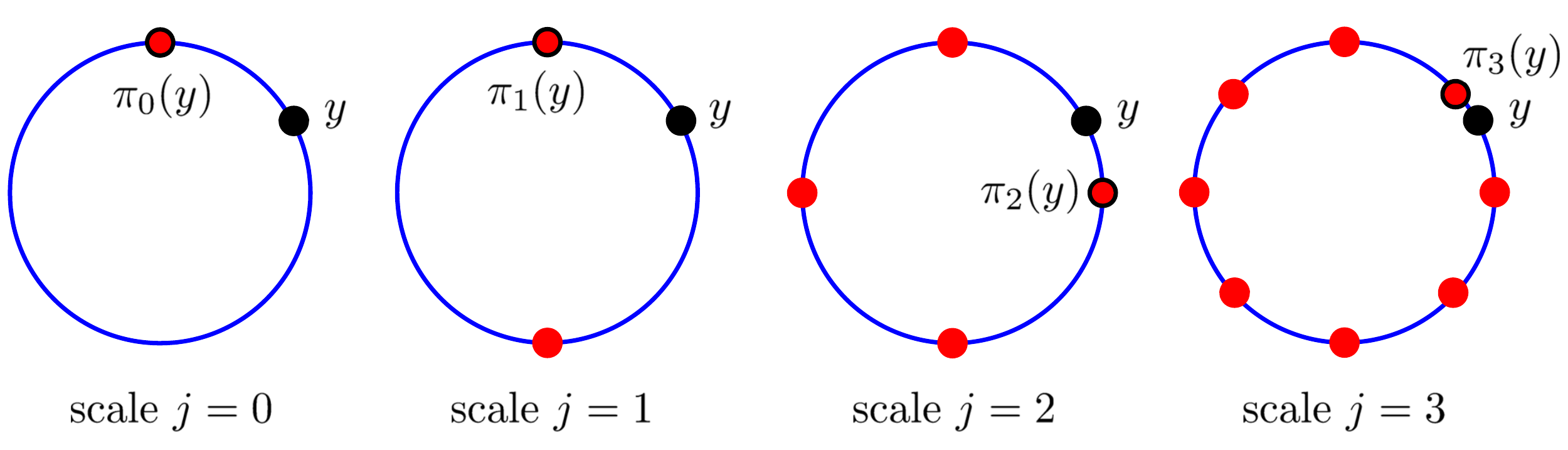}
\end{center}
\caption{\small\sl A sequence of increasingly finer covers for the unit circle. Also shown is an arbitrary point $y$ and its projection $\pi_j(y)$ onto each cover.}
\label{fig:Example}
\end{figure}
We can use the telescoping sum 
$$
y =\pi_0(y)+ \sum_{j\ge 1} (\pi_{j}(y)-\pi_{j-1}(y))
$$
to write that
\begin{align}\label{eq:example}
		& \pr{\sup_{x_1,x_2\in\manifold} \frac{ \|\Phi x_1-\Phi x_2\|}{\|x_1-x_2\|}>1+\epsilon}\nonumber\\
		& = \pr{\sup_{y\in U(\manifold)} \|\Phi y\|>1+\epsilon}\nonumber\\
		& = \pr{\sup_{y\in \manifold} \|\Phi y\|>1+\epsilon}\nonumber\\
		& \le \pr{\sup_{{y\in\manifold}}\|\Phi \pi_0(y)\|+\sum_{j\ge1} \sup_{y\in\manifold} \|\Phi(\pi_{j}(y)-\pi_{j-1}(y))\|>1+\sum_{j\ge0} \epsilon_j}\nonumber\\
		& \le \pr{\max_{{p\in C_0}}\|\Phi p\|+\sum_{j\ge1} \max_{(p,q)\in C_{j}\times C_{j-1}} \|\Phi  (p-q)\|>1+\sum_{j\ge0} \epsilon_j}\nonumber\\
		& \le \#C_0\cdot\max_{p\in C_0} \pr{\|\Phi p\|>1+\epsilon_0}	+ \sum_{j\ge1} \#C_{j}\cdot\#C_{j-1} \max_{(p,q)\in C_{j+1}\times C_j} \pr{\|\Phi (p-q)\|>\epsilon_j},	
\end{align}
where $\{\epsilon_j\}$ is an exponentially-fast decreasing sequence of constants such that $\sum_j \epsilon_j=\epsilon$. The third line uses the fact that $U(\manifold)=\manifold$ here. The last line above uses the union bound. Therefore the failure probability of obtaining a stable embedding of $\manifold$ is controlled by an infinite series that only involves the sequence of covers constructed earlier. As we will see in more detail later, given enough measurements, the (exponentially growing) size of the covers $\#C_j$ can be balanced by the (exponentially decreasing) failure probabilities in the last line of \eqref{eq:example} to guarantee that overall failure probability is exponentially small. A more general version of this chaining argument is detailed in Appendix~\ref{sub:Applying-the-Chaining arg}.

\section{Manifold-based signal recovery and parameter estimation}
\label{sec:srpe}

In this section, we establish theoretical bounds on the accuracy to which problems Q1 (signal recovery) and Q2 (parameter estimation) can be solved under a manifold model. To be specific, let us consider a length-$\dim$ signal $x$ that is not necessarily $K$-sparse, but rather that we assume lives on or near some known $\mdim$-dimensional manifold $\manifold \subset \real^\dim$. From a collection of measurements
$$y = \proj x + \noise,$$
where $\proj$ is a random $\pdim \times \dim$ matrix and $\noise \in
\real^\pdim$ is an additive noise vector, we would like to recover
either $x$ or a parameter $\theta$ that generates $x$.

\notechecked{edited this paragraph to remove assumption that minimizer is unique} For the signal recovery problem, we will consider the following as a method for estimating $x$. Solve the program
\begin{equation}\label{eq:csmrecover}
\min_{x'\in\mathcal{M}}\|y-\Phi x'\|,
\end{equation}
and let $\widehat{x}$, as an estimate of $x$, be a solution of the above program. We also let $x^{*}$ be a solution of the program
\begin{equation}
\min_{x'\in\mathcal{M}}\|x-x'\|\label{eq:csmopt}
\end{equation}
and, therefore, an optimal “nearest neighbor” to $x$ on $\mathcal{M}$. To consider signal recovery successful, we would like to guarantee that $\norm{x-\widehat{x}}$ is not much larger than $\norm{x-x^\ast}$.

\notechecked{edited this paragraph to remove assumption that minimizer is unique} For the parameter estimation problem, where we presume $x \approx
x_\theta$ for some $\theta \in \Theta$, we propose a similar method for estimating $\theta$ from the compressive measurements. Solve the program
\begin{equation}
\min_{\theta'\in\Theta}\|y-\Phi x_{\theta'}\|,\label{eq:parmrecover}
\end{equation}
 and let $\widehat{\theta}$, as an estimate of $\theta$, be a solution of the above program. Also let $\theta^{*}$ be a solution of the program
\begin{equation}
 \min_{\theta'\in\Theta}\|x-x_{\theta'}\|.\label{eq:parmopt}
\end{equation}
Here, $\theta^{*}$ is an “optimal estimate” that could be obtained using the full data $x\in\mathbb{R}^{N}$. (If $x = x_\theta$ exactly for some $\theta$, then $\theta^\ast =
\theta$; otherwise this formulation allows us to consider signals
$x$ that are not precisely on the manifold $\manifold$ in
$\real^\dim$. This generalization has practical relevance; a local
image block, for example, may only approximately resemble a straight
edge, which has a simple parameterization.) To consider parameter
estimation successful, we would like to guarantee that
$d_\Theta(\widehat{\theta},\theta^\ast)$ is small.

As we will see, bounds pertaining to accurate signal recovery can
often be extended to imply accurate parameter estimation as well.
However, the relationships between distance $d_\Theta$ in parameter
space and distances $d_\manifold$ and $\|\cdot\|$ in the signal
space can vary depending on the parametric signal model under study.
Thus, for the parameter estimation problem, our ability to provide
generic bounds on $d_\Theta(\widehat{\theta},\theta^\ast)$ will be
restricted. In this paper we focus primarily on the signal recovery
problem and provide preliminary results for the parameter estimation
problem that pertain most strongly to the case of isometric
parameterizations.

In this paper, we do not confront in depth the question of how a recovery program such as (\ref{eq:csmrecover}) can be efficiently solved. Some efforts in this direction have recently appeared in~\cite{shah2011iterative,iwen2012approximation,hegde2012signal}.
In~\cite{shah2011iterative}, the authors guarantee the success of a gradient-projection algorithm for recovering a signal that lives exactly on the manifold from noisy compressive measurements. The keys to the success of this method are a stable embedding of the manifold (as is guaranteed by~\cite{baraniuk2009random} or our Theorem~\ref{thm:new manifold embedding}) and the knowledge of the projection operator onto the manifold within $\real^N$.
In~\cite{iwen2012approximation}, the authors construct a geometric multi-resolution approximation of a manifold using a collection of affine subspaces. A major contribution of that work is a recovery algorithm that works by assigning a measured signal to the closest projected affine subspace in the compressed domain. Two recovery results are presented. In the first of these, the number of measurements is independent of the ambient dimension and the recovery error holds for any given signal in the ambient space. All of this is analogous to our Theorem~\ref{theo:bound3} (a probabilistic instance-optimal bound in $\ell_2$), but the recovery is guaranteed for a particular algorithm. Unlike that result, however, our Theorem~\ref{theo:bound3} includes explicit constants, allows for the consideration of measurement noise, and falls nearly for free out of our novel analytical framework based on chaining. A second result appearing in~\cite{iwen2012approximation} provides a special type of deterministic instance-optimal bound for signal recovery and involves embedding arguments that extend those in~\cite{baraniuk2009random}. It would be interesting to see if our improved embedding guarantees in the present paper could now be used to remove the dependence on the ambient dimension appearing in that result.
In~\cite{chen2010compressive}, the authors provide a Bayesian treatment of the signal recovery problem using a mixture of (low-rank) Gaussians for approximating the manifold.
Furthermore, some discussion of signal recovery is provided in~\cite{baraniuk2009random}, with application-specific examples provided in~\cite{mbwPhdThesis,davenport2007sfc}. Unfortunately, it is difficult to propose a single general-purpose algorithm for solving (\ref{eq:csmrecover}) in $\real^\pdim$, as even the problem (\ref{eq:csmopt}) in $\real^\dim$ may be difficult to solve depending on certain nuances (such as topology) of the individual
manifold.
Additional complications arise when the manifold $\manifold$ is non-differentiable, as may happen when the signals $x$ represent 2-D images. However, just as a multiscale regularization can be incorporated into Newton's method for solving (\ref{eq:csmopt}) (see~\cite{WakinSPIEiam}), an analogous regularization can be
incorporated into a compressive measurement operator $\proj$ to facilitate Newton's method for solving (\ref{eq:csmrecover}) (see~\cite{donohoECS,mbwPhdThesis,msSmashed}). For manifolds that lack differentiability, additional care must be taken when applying results such as Theorem~\ref{theo:manifoldjl}.
We therefore expect that the research on signal recovery and approximation based on low-dimensional manifold models will witness even more growth in the future.

It is  also crucial to study the theoretical limits and guarantees in this problem; in what follows, we will consider both deterministic and
probabilistic instance-optimal bounds for signal recovery and
parameter estimation, and we will draw comparisons to the
sparsity-based CS results of Section~\ref{sec:csresults}. Our bounds
are formulated in terms of generic properties of the manifold (as
mentioned in Section~\ref{sec:mmodels}), which will vary from signal
model to signal model. In some cases, calculating these may be
possible, whereas in other cases it may not. Nonetheless, we feel
the results in this paper highlight the relative importance of these
properties in determining the requisite number of measurements.

\subsection{A deterministic instance-optimal bound in $\ell_2$}
\label{sec:deter}

\notechecked{revised paragraph to mention noise} We begin by seeking an instance-optimal bound. That is,
for a measurement matrix $\proj$ that meets (\ref{eq:mmain}) for all
$x_1, x_2 \in \manifold$, we seek an upper bound for the relative
reconstruction error
$$
\frac{\norm{x-\widehat{x}}}{\norm{x-x^\ast}}
$$
that holds uniformly for all $x \in \real^\dim$. We would also like this bound to account for noise in the measurements. In this section we
consider only the signal recovery problem; however, similar bounds would apply to parameter estimation. We have the following result, which is proved in Appendix~\ref{app:kappa}.

\begin{thm}
\notechecked{adjusted theorem to account for noise} Fix $0 <\eps \le 1/3$ and $0 < \rho < 1$. Let $\manifold$ be as described in Theorem~\ref{thm:new manifold embedding}. Assume that $\proj$ satisfies \eqref{eq:mmain} for all pairs of points $x_1,x_2\in\manifold$. Take
$x \in \real^\dim$, let $y = \proj x+n$, and let the recovered estimate $\widehat{x}$ and \notechecked{replaced ``the'' with ``an'' so as not to imply minimizer is unique} an optimal estimate $x^\ast$ be as defined in (\ref{eq:csmrecover}) and (\ref{eq:csmopt}). Then the following holds:
\begin{equation}
\norm{x-\widehat{x}} \le
 (1+2\epsilon)\left(2\sing_M(\proj)+1\right)\norm{x-x^\ast}+(2+4\epsilon)\|n\|,\label{eq:kappabound}
\end{equation}
where $\sing_M(\proj)$ is the largest singular value of $\proj$.
\label{thm:kappa}
\end{thm}
In particular, it is interesting to consider the case where $\proj$ is a random Gaussian matrix as described in Theorem~\ref{thm:new manifold embedding}.
It is well-known, e.g., \cite[Corollary 5.35]{vershynin2010introduction}, that the nonzero singular values of $\proj$ satisfy the following:
\begin{equation}
	\label{eq:1st bound on spec norm of gaussian matrix}
	\pr{\sing_M(\proj)>\sqrt{\frac{N}{M}}+1+t}\le e^{-t^2M/2},
\end{equation}
\begin{equation}
	\label{eq:2nd bound on spec norm of gaussian matrix}
	\pr{\sing_m(\proj)<\sqrt{\frac{N}{M}}-1-t}\le e^{-t^2M/2},
\end{equation}
for $t>0$. Here, $\sing_M(\proj)$ and $\sing_m(\proj)$ are the largest and smallest (nonzero) singular values of $\proj$, respectively. Suppose that $M$ satisfies \eqref{eq:mmeasmain} so that the promises of Theorem~\ref{thm:new manifold embedding} hold except for a probability of at most $\rho$. Set $t=1$ in \eqref{eq:1st bound on spec norm of gaussian matrix}. Now since $e^{-M/2}\le \rho$, we have that
$$
\sing_M(\proj)\le \sqrt[•]{\frac{N}{M}}+2,
$$
except with a probability of at most $\rho$. In combination with Theorem~\ref{thm:kappa}, it finally follows that, except with a probability of at most $2\rho$, $\proj$ satisfies \eqref{eq:mmain} for every pair of points on the manifold and that \notechecked{edited remainder of this paragraph to account for the revisions to Theorem~\ref{thm:kappa}.}
\begin{equation}
	\norm{x-\widehat{x}} \le
 (1+2\epsilon)\left(2\sqrt{\frac{N}{M}}+5\right)\norm{x-x^\ast}+(2+4\epsilon)\|n\|,\label{eq:kappabound}
\end{equation}
for every $x\in\real^N$.  Here, $\widehat{x}$ and $x^*$  are as defined in \eqref{eq:csmrecover} and  \eqref{eq:csmopt}. In the noise-free case ($\|n\|=0$) and  as $\frac{\pdim}{\dim} \rightarrow 0$, the bound on the right hand side of (\ref{eq:kappabound}) grows as $(2+4\epsilon)\sqrt{\frac{\dim}{\pdim}}$. Unfortunately, this is not desirable for signal recovery. Supposing, for example, that we wish to ensure $\norm{x-\widehat{x}} \le \KappaOne \norm{x-x^\ast}$ for all $x \in \real^\dim$ (assuming no measurement noise), then using the bound (\ref{eq:kappabound}) we would require that $\pdim \ge \KappaTwo \dim$ regardless of the dimension $\mdim$ of the manifold.

The weakness of this bound is a geometric necessity; indeed, the bound itself is quite tight in general, as the following simple example illustrates. The proof can be found in Appendix~\ref{sec:Proof of l2 l2 bound is sharp}.

\begin{prop} \label{lemma:l2 l2 bound is sharp}
Fix $0<\epsilon\le1/3$. Let $\manifold$ denote the line segment in $\real^\dim$ joining the origin and $e_1:=[1,0,0,\dots,0]^T$. Suppose that $\proj$ satisfies \eqref{eq:mmain} for all $x_1,x_2\in\manifold$ and that $\sing_m(\proj)\ge 8/3$. \notechecked{edited remainder of this proposition to clarify there is no measurement noise.} Then, there exists a point $x\in\real^N$ such that if $y=\Phi x$ (with no measurement noise), and if $\widehat{x}$ and $x^\ast$ are defined in \eqref{eq:csmrecover} and \eqref{eq:csmopt},
$$
\frac{\|x-\widehat{x}\|}{\|x-x^\ast\|}\ge \frac{1}{2(1+\epsilon)}\sing_m(\proj).
$$
\end{prop}

In particular, consider the case where $\proj$ is a random Gaussian matrix as described in Theorem~\ref{thm:new manifold embedding}. According to \eqref{eq:COM_norm} and \eqref{eq:2nd bound on spec norm of gaussian matrix} (with $t=1$), the following two statements are valid except with a probability of at most $2e^{-M\epsilon^2/6}+2e^{-M/2}\le 4e^{-M\epsilon^2/6}$. First, \eqref{eq:mmain} holds for every $x_1,x_2\in\manifold$. Second,
\begin{equation}\label{eq:2nd statement}
	\sing_m(\proj)\ge \sqrt{\frac{N}{M}}-2.
\end{equation}
If we assume that $N/M\ge(14/3)^2$, we can conclude, using Proposition~\ref{lemma:l2 l2 bound is sharp}, that $\proj$ satisfies \eqref{eq:mmain} for every $x_1,x_2\in\manifold$ and yet there exists $x\in\real^N$ such that  $$
\frac{\|x-\widehat{x}\|}{\|x-x^\ast\|}\ge \frac{1}{4(1+\epsilon)}\sqrt{\frac{N}{M}},
$$
except with a probability of at most $4e^{-M\epsilon^2/6}$ on the choice of $\proj$.

It is worth recalling that, as we discussed in
Section~\ref{sec:csresults}, similar difficulties arise in
sparsity-based CS when attempting to establish a deterministic
$\ell_2$ instance-optimal bound. In particular, to ensure that
$\norm{x-\widehat{x}} \le \CddOne \norm{x-x_\sparsity}$ for all $x
\in \real^\dim$, it is known~\cite{CDD} that this requires $\pdim
\ge \CddTwo \dim$ regardless of the sparsity level $\sparsity$.

In sparsity-based CS, there have been at least two types of
alternative approaches. The first are the deterministic
``mixed-norm'' results of the type given in (\ref{eq:mixed2}). These involve the use of an alternative norm such
as the $\ell_1$ norm to measure the distance from the coefficient
vector $\alpha$ to its best $\sparsity$-term approximation
$\alpha_\sparsity$. While it may be possible to pursue similar
directions for manifold-modeled signals, we feel this is undesirable
as a general approach because when sparsity is no longer part of the
modeling framework, the $\ell_1$ norm has less of a natural meaning.
Instead, we prefer to seek bounds using $\ell_2$, as that is the
most conventional norm used in signal processing to measure energy
and error.

Thus, the second type of alternative bounds in sparsity-based CS
have involved $\ell_2$ bounds in probability, as we discussed in
Section~\ref{sec:csresults}. Indeed, the performance of both
sparsity-based and manifold-based CS is often much better in
practice than a deterministic $\ell_2$ instance-optimal bound might
indicate. The reason is that, for any $\proj$, such bounds consider
the {\em worst case} signal over all possible $x \in \real^\dim$.
Fortunately, this worst case is not typical. As a result, it is
possible to derive much stronger results that consider any given
signal $x \in \real^\dim$ and establish that for most random
$\proj$, the recovery error of that signal $x$ will be small.

\subsection{Probabilistic instance-optimal bounds in $\ell_2$}
\label{sec:prob}

For a given measurement operator $\proj$, our bound in
Theorem~\ref{thm:kappa} applies uniformly to any signal in
$\real^\dim$. However, a much sharper bound can be obtained by
relaxing the deterministic requirement.

\subsubsection{Signal recovery}

Our first bound applies to the signal recovery problem. The proof of this result is provided in Appendix~\ref{sec:proof of prob par est theorem} and, like that of Theorem~\ref{theo:manifoldjl}, involves a generic chaining argument.

\begin{thm}
	\notechecked{edited theorem to provide two bounding terms involving $x^{*}$, one of which goes to $0$ as $x^{*} \rightarrow x$.}
    Suppose $x \in \real^\dim$. Let $\manifold$ be a compact $\mdim$-dimensional Riemannian submanifold of $\real^\dim$ having condition number $1/\condition$ and volume $\volume$. Conveniently assume that\footnote{\notechecked{revised} Theorem~\ref{theo:bound3} still holds, with worse (larger) constants, after relaxing this assumption.}
	\begin{equation}
	      \frac{\volume}{\condition^{K}}\ge \left(\frac{21}{\sqrt{K}}\right)^{K}.\label{eq:conv assump for par est}
	\end{equation}
	Fix $0 < \eps \le 1/3$ and $0 < \rho < 1$. Let $\proj$ be an $\pdim \times \dim$ random matrix populated with i.i.d.\ zero-mean Gaussian random variables with variance $1/M$, chosen independently of $x$, with $M$ satisfying \eqref{eq:mmeasmain}.
%
%
	Let $\noise \in \real^\pdim$, let $y = \proj x + \noise$, and let
	the recovered estimate $\widehat{x}$ and \notechecked{replaced ``the'' with ``an'' so as not to imply minimizer is unique} an optimal estimate
	$x^\ast$ be as defined in (\ref{eq:csmrecover}) and
	(\ref{eq:csmopt}). Then with a probability of at
	least $1-4\rho$, the following statement holds:
	\begin{equation}
		\norm{x-\widehat{x}} \le
		\min\left(  \left(1+3\epsilon\right)\left\Vert x-x^{*}\right\Vert +\frac{\epsilon\condition}{40}\,,\,  (1+2\epsilon)\left(2\sqrt{\frac{N}{M}}+5\right)\|x-x^*\| \right) +\left(2+4\epsilon\right)\left\Vert \noise\right\Vert.
		\label{eq:bound3}
	\end{equation}
	\label{theo:bound3}
\end{thm}

\notechecked{next paragraph added to discuss revised theorem}
Roughly speaking, one can discern two different operating regimes in \eqref{eq:bound3}:
\begin{itemize}
\item When $x$ is sufficiently far from the manifold ($\|x-x^*\| \gg \epsilon \tau$), then \eqref{eq:bound3} roughly reads
$$
\norm{x-\widehat{x}} \le (1 + 3\eps)\norm{x-x^\ast}+
		(2+4\eps)\norm{\noise}.
$$
In particular, by setting $\|n\|=0$ in the bound above (which corresponds to the noise-free setup), we obtain a multiplicative error bound: The recovery error from compressive measurements $\norm{x-\widehat{x}}$ is no larger than twice the recovery error from a full set of measurements $\|x-x^*\|$.
\item On the other hand, when $x$ is close to the manifold ($\|x-x^*\| \ll \epsilon \tau \sqrt{M/N}$), then \eqref{eq:bound3} becomes
$$
\norm{x-\widehat{x}} \le (1+2\epsilon)\left(2\sqrt{\frac{N}{M}}+5\right)\norm{x-x^\ast} +
		(2+4\eps)\norm{\noise}. 
$$
When $\|n\|=0$, we still obtain a multiplicative error bound but with a larger factor in front of $\|x-x^*\|$.
\end{itemize}

Let us also compare and contrast our bound with the analogous results for sparsity-based CS. Like Theorem~\ref{theo:mixed2}, we consider the problem of signal recovery in the presence of additive measurement noise. Both bounds relate the recovery error $\norm{x-\widehat{x}}$ to the proximity of $x$ to its nearest
neighbor in the concise model class (either $x_\sparsity$ or $x^\ast$ depending on the model), and both bounds relate the recovery error $\norm{x-\widehat{x}}$ to the amount $\norm{\noise}$ of additive measurement noise.
However, Theorem~\ref{theo:mixed2} is a deterministic bound whereas Theorem~\ref{theo:bound3} is probabilistic, and our bound (\ref{eq:bound3}) measures proximity to the concise model in the $\ell_2$ norm, whereas (\ref{eq:mixed2}) uses the $\ell_1$ norm.

Our bound can also be compared with (\ref{eq:devore}), as both are instance-optimal bounds in probability, and both use the $\ell_2$ norm to measure proximity to the concise model. However, we note that unlike (\ref{eq:devore}), our bound (\ref{eq:bound3}) allows the consideration of measurement noise.

\subsubsection{Parameter estimation}

Above we have derived a bound for the signal recovery problem, with
an error metric that measures the discrepancy between the recovered
signal $\widehat{x}$ and the original signal $x$.

However, in some applications it may be the case that the original
signal $x \approx x_{\theta^\ast}$, where $\theta^\ast \in \Theta$
is a parameter of interest. In this case we may be interested in
using the compressive measurements $y = \proj x + \noise$ to solve
the problem (\ref{eq:parmrecover}) and recover an estimate
$\widehat{\theta}$ of the underlying parameter.

Of course, these two problems are closely related. However, we
should emphasize that guaranteeing $\norm{x-\widehat{x}} \approx
\norm{x-x^\ast}$ does not automatically guarantee that
$\gdist(x_{\widehat{\theta}},x_{\theta^\ast})$ is small (and
therefore does not ensure that
$d_\Theta(\widehat{\theta},\theta^\ast)$ is small). If the manifold
is shaped like a horseshoe, for example, then it could be the case
that $x_{\theta^\ast}$ sits at the end of one arm but
$x_{\widehat{\theta}}$ sits at the end of the opposing arm. These
two points would be much closer in a Euclidean metric than in a
geodesic one.

Consequently, in order to establish bounds relevant for parameter
estimation, our concern focuses on guaranteeing that the geodesic
distance $\gdist(x_{\widehat{\theta}},x_{\theta^\ast})$ is itself
small. Our next result is proved in Appendix~\ref{app:bound4}.

\begin{thm}\notechecked{edited the theorem} Suppose $x \in \real^\dim$, and fix $0 < \eps \le 1/3$ and $0 < \rho < 1$. Let $\manifold$ and $\proj$ be as described in Theorem~\ref{theo:bound3}, assuming that $M$ satisfies \eqref{eq:mmeasmain} and that the convenient assumption \notechecked{corrected equation reference} \eqref{eq:conv assump for par est} holds. Let $\noise \in \real^\pdim$, let $y = \proj x + \noise$, and let the recovered estimate $\widehat{x}$ and  \notechecked{replaced ``the'' with ``an'' so as not to imply minimizer is unique} an optimal estimate $x^\ast$ be as defined in (\ref{eq:csmrecover}) and
(\ref{eq:csmopt}). If $\norm{x-x^\ast}+\frac{10}{9}\norm{\noise}\le 0.163\condition$, then with probability at least $1-4\rho$ the following statement holds:
\begin{equation}
\gdist(\widehat{x}, x^\ast) \le
\min\left( (4+6\epsilon)\|x-x^*\|+\frac{\epsilon\tau}{20}\,,\,\left((4+8\epsilon)\sqrt{\frac{N}{M}}+12+20\epsilon\right)\|x-x^*\| \right) +(4+8\epsilon)\|n\|.
\label{eq:bound4}
\end{equation}
\label{theo:bound4}
\end{thm}

In several ways, this bound is similar to (\ref{eq:bound3}). Both
bounds relate the recovery error to the proximity of $x$ to its
nearest neighbor $x^\ast$ on the manifold and to the amount
$\norm{\noise}$ of additive measurement noise.
Both bounds also have an additive term on the right hand side that
is small in relation to the condition number $\condition$.

In contrast, (\ref{eq:bound4}) guarantees that the recovered
estimate $\widehat{x}$ is near to the optimal estimate $x^\ast$ in
terms of geodesic distance along the manifold. Establishing this
condition required the additional assumption that \notechecked{updated inequality} $\norm{x-x^\ast}+\frac{10}{9}\norm{\noise}\le 0.163\condition$. Because $\condition$ relates to
the degree to which the manifold can curve back upon itself at long
geodesic distances, this assumption prevents exactly the type of
``horseshoe'' problem that was mentioned above, where it may happen
that $\gdist(\widehat{x}, x^\ast) \gg \norm{\widehat{x}-x^\ast}$.
Suppose, for example, it were to happen that $\norm{x-x^\ast}
\approx \condition$ and $x$ was approximately equidistant from both
ends of the horseshoe; a small distortion of distances under $\proj$
could then lead to an estimate $\widehat{x}$ for which
$\norm{x-\widehat{x}} \approx \norm{x-x^\ast}$ but
$\gdist(\widehat{x},x^\ast) \gg 0$. Similarly, additive noise could
cause a similar problem of ``crossing over'' in the measurement
space. Although our bound provides no guarantee in these situations,
we stress that under these circumstances, accurate parameter
estimation would be difficult (or perhaps even unimportant) in the
original signal space $\real^\dim$.

Finally, we revisit the situation where the original signal $x
\approx x_{\theta^\ast}$ for some $\theta^\ast \in \Theta$ (with
$\theta^\ast$ satisfying (\ref{eq:parmopt})), where the measurements
$y = \proj x + \noise$, and where the recovered estimate
$\widehat{\theta}$ satisfies (\ref{eq:parmrecover}). We consider the
question of whether (\ref{eq:bound4}) can be translated into a bound
on $d_\Theta(\widehat{\theta},\theta^\ast)$. As described in
Section~\ref{sec:mmodels}, in signal models where $\manifold$ is
isometric to $\Theta$, this is automatic: we have simply that
\begin{equation*}
\gdist(x_{\widehat{\theta}},x_{\theta^\ast}) =
d_\Theta(\widehat{\theta},\theta^\ast).
\end{equation*}
Such signal models are not nonexistent. Work by Donoho and
Grimes~\cite{DonohoGrimesISOMAP}, for example, has characterized a
variety of articulated image classes for which~(\ref{eq:isometry})
holds or for which $\gdist(x_{\theta_1},x_{\theta_2}) = \GrimesOne
d_\Theta(\theta_1, \theta_2)$ for some constant $\GrimesOne > 0$. In
other models it may hold that
$$
\GrimesTwo \gdist(x_{\theta_1},x_{\theta_2}) \le d_\Theta(\theta_1,
\theta_2) \le \GrimesThree \gdist(x_{\theta_1},x_{\theta_2})
$$
for constants $\GrimesTwo, \GrimesThree > 0$. Each of these
relationships may be incorporated to the bound (\ref{eq:bound4}).

\section{Conclusions}
\label{sec:concl}

In this paper, we have provided an improved and non-asymptotic lower bound on the number of requisite measurements to ensure a stable embedding of a manifold under a random linear measurement operator. We have also considered the tasks of signal recovery and parameter estimation using compressive measurements of a
manifold-modeled signal, and we have established theoretical bounds on the accuracy to which these questions may be answered. Although these problems differ substantially from the mechanics of sparsity-based CS, we have seen a number of similarities that arise due to the low-dimensional geometry of the
each of the concise models.
First, we have seen that a sufficient number of compressive measurements can guarantee a stable embedding of either type of signal family, and the requisite number of measurements scales essentially linearly with the information level of the signal.
Second, we have seen that deterministic instance-optimal bounds in $\ell_2$ are necessarily weak for both problems.
Third, we have seen that probabilistic instance-optimal bounds in $\ell_2$ can be derived that give the optimal scaling with respect to the signal proximity to the concise model and with respect to the amount of measurement noise.
Thus, our work supports the growing evidence that manifold-based models can be used with high accuracy in compressive signal processing.

Most of our analysis in this paper rests on a new analytical framework for studying manifold embeddings that uses tools from the theory of empirical processes (namely, the idea of generic chaining). While such tools are becoming more widely adopted in the analysis of sparsity-based CS problems, we believe they are also very promising for studying the interactions of nonlinear signal families (such as manifolds) with random, compressive measurement operators. We hope that the chaining argument in this paper will be useful for future investigations along these lines.

\section*{Acknowledgements}

M.B.W.\ is grateful to Rich Baraniuk and the Rice CS research team for many stimulating discussions. A.E. thanks Justin Romberg for introducing him to the generic chaining and other topics in the theory of empirical processes, Han Lun Yap for his valuable contributions to an early version of the proof of Theorem~\ref{thm:new manifold embedding} and many productive discussions about the topic, and Alejandro Weinstein for helpful discussions. Finally, both authors would like to acknowledge the tremendous and positive influence that the late Partha Niyogi has had on our work.

\appendix

\section{Toolbox}
\label{sec:useful diff geom reuslts}

We begin by introducing some notation that will be used throughout the rest of the appendices.

In this paper, $\ints$ stands for the set of nonnegative integers. The tangent space of $\manifold$ at $p\in\manifold$ is denoted $\tan_{p}$. The orthogonal projection operator onto this linear subspace is  denoted by $\downarrow_{p}$. We let $\angle\left[\cdot,\cdot\right]$ represent the angle between two
vectors after being shifted to the same starting point. Throughout this paper, $d_{\manifold}\left(\cdot,\cdot\right)$ measures the geodesic distance between two points on $\manifold$. By {\em $r$-ball} we refer to a Euclidean (open) ball of radius $r>0$. In addition, with $\ball_N$ we denote the unit ball in $\real^N$ with volume $V_{\ball_N}$ and we reserve $\ball_N(p,r)$ to represent an $N$-dimensional $r$-ball centered at $p$ in $\real^N$. For $r>0$, let $\intersection_\manifold(p,r):=\mathcal{M}\cap \ball_N(p,r)-p$ denote a (relatively) open neighborhood of $p$ on $\mathcal{M}$ after being shifted to the origin. Here the subtraction is in the Minkowski sense. \notechecked{clarified next sentence} The $K$-dimensional ball of radius $r$ in $\mathcal{T}_{p}$ will be denoted by $\mathcal{B}_{\mathcal{T}_{p}}$; this ball is centered at the origin, as $\tan_{p}$ is a linear subspace. Unless otherwise stated, all distances are measured in the Euclidean metric.

A collection of $N$-dimensional $r$-balls that covers $\manifold$ is called an $r$-\emph{cover} for $\manifold$, with their centers forming \notechecked{added ``so-called'' here} a so-called $r$-\emph{net} for $\manifold$. Notice that in general we do not require a net for $\manifold$ to be a subset of $\manifold$. 
However, we define the {\em covering number} of $\manifold$ at scale $r$, $\coverno_\manifold(r)$, to be the cardinality of a minimal $r$-net for $\manifold$ among all subsets of $\manifold$. (In other words, $\coverno_\manifold(r)$ is the smallest number of $r$-balls centered on $\manifold$ that it takes to cover $\manifold$.) A maximal $r$-separated subset of $\manifold$ is called an $r$-\emph{packing} for $\manifold$.
The \emph{packing number} of $\manifold$ at scale $r>0$, denoted by $\packing_\manifold(r)$, is the cardinality of such a set. It can be easily verified that an $r$-packing for $\manifold$ is also an $r$-cover for $\manifold$, so
\begin{equation}
	\coverno_\manifold(r)\le \packing_\manifold(r). \label{eq:coverno n packing}
\end{equation}

\notechecked{clarified concepts about subspaces/operators in this paragraph} The concept of (principal) angle between subspaces will later come in handy. The (principal) angle between two linear subspaces $\mathcal{T}_{p}$ and $\mathcal{T}_{q}$ is defined such that $\cos\left(\angle\left[\mathcal{T}_{p},\mathcal{T}_{q}\right]\right)=\min_{u}\max_{v}\left|\left\langle u,v\right\rangle \right|$, where the unit vectors $u$ and $v$ belong to $\mathcal{T}_{p}$   and $\mathcal{T}_{q}$, respectively. It is known that
\begin{equation}
\|\downarrow_{p}(\cdot)-\downarrow_{q}(\cdot)\|_{2,2}=\sin\left(\angle\left[\mathcal{T}_{p},\mathcal{T}_{q}\right]\right),
\label{eq:principal angle and projection norms}
\end{equation}
where, as defined above, $\downarrow_{p}$ and $\downarrow_{q}$ are linear orthogonal projectors onto the tangent subspaces $\mathcal{T}_{p}$
and $\mathcal{T}_{q}$, respectively~\cite[Theorem 2.5]{subspacePerturbation}.\footnote{In fact, \eqref{eq:principal angle and projection norms} holds for any two linear subspaces (not only tangent subspaces of a manifold).} The norm above is the spectral norm, namely the operator norm from $\mathbb{R}^{N}$ equipped with $\ell_{2}$ to itself.

We will also use the following conventions to clarify the exposition. For $x_1\ne x_2\in\real^N$, define
\begin{equation*}
	 U\left(x_{1},x_{2}\right):=\frac{x_{2}-x_{1}}{\left\Vert x_{2}-x_{1}\right\Vert }.
\end{equation*}
Additionally, we let $U(S_1,S_2)$ denote the set of directions of all the chords connecting two sets $S_{1},S_{2}\subseteq\real^N$, namely
\[
	U(S_{1},S_{2}):=\left\{ U\left(x_{1},x_{2}\right)\,:\, x_{1}\in S_{1},x_{2}\in S_{2},x_1\ne x_2\right\}.
\]
Clearly, $U(S_{1},S_{2})\subseteq\mathbb{S}^{N-1}$, where $\mathbb{S}^{N-1}$ is the unit sphere in $\mathbb{R}^{N}$.
Whenever possible, we also simplify our notation by using $U\left(S\right):=U(S,S)$.

Below we list a few useful results (mostly from differential geometry)  which are used throughout the rest of the paper. We begin with a well-known bound on the covering number of Euclidean balls, e.g., \cite[Lemma 5.2]{vershynin2010introduction}.\footnote{Lemma~5.2 in~\cite{vershynin2010introduction} concerns the unit sphere in $\real^K$, but the result still holds for the unit Euclidean ball using essentially the same argument.}
\begin{lemma}\label{lemma:cover no of unit ball}
	A $K$-dimensional unit ball can be covered by at most $(3/r)^K$ $r$-balls with $r\le1$.
\end{lemma}
We now recall several results from Sections 5 and 6 in \cite{niyogi}. Unfortunately we were unable to confirm for ourselves some of the original proofs appearing in~\cite{niyogi}. Therefore, some of the statements and proofs below differ slightly from their original counterparts. The first result is closely related to Lemma 5.3 in \cite{niyogi}.
\begin{lemma}\label{lemma:niyogi geom result}
    Fix $p,q\in \manifold$, such that $\left\Vert q-p\right\Vert<2\condition$. Then $\angle\left[q-p,\downarrow_{p}(q-p)\right]\le\sin^{-1}\left(\Vert q-p\Vert/2\condition\right)$. \notechecked{changed $<$ to $\le$ in the angle bound}
\end{lemma}
\begin{proof}
    \notechecked{edited this proof (and added figure) for clarity} Consider the unit vector $v$ along $(q-p)-\downarrow_p(q-p)\perp \tan_p$ and the point $z:=p+\condition\cdot v$. Observe that $z-p$ is orthogonal to the manifold at $p$. By definition of the condition number, the distance from $z$ to the manifold is minimized at $p$ and we must therefore have $\Vert z-q\Vert \ge \condition$.\footnote{\notechecked{footnote is edited.} To see this, consider a sequence of points $z_n :=p+(\condition-1/n) \cdot v$ for integer values of $n$. For each $n$, $\|z_n - p\| = \condition-1/n < \condition$, and $z_n-p$ is orthogonal to the manifold at $p$. Therefore, by the definition of the condition number, no point $q' \in \manifold$ can satisfy $\| z_n - q' \| < \|z_n - p\|$. Thus, the distance from $z_n$ to the manifold equals $\condition-1/n$. Taking the limit as $n \rightarrow \infty$ and using the continuity of the distance function, we conclude that the distance from $z$ to $\manifold$ equals $\condition$. So, no point $q' \in \manifold$ can satisfy $\| z - q' \| < \|z - p\| = \condition$.} Now consider the triangle formed by the points $p,q,z$ and the line $l$ passing through $z$ and perpendicular to $q-p$. Let $z'$ denote the intersection of $l$ with the line passing through $p$ and $q$. (See Figure~\ref{fig:Proof picture}.)
\begin{figure}[H]
	\centering
	  \includegraphics[scale=.60]{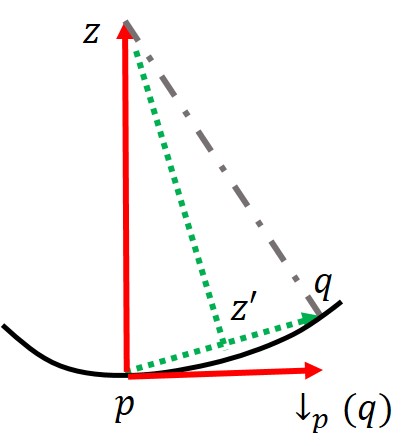}
	  \caption{\small\sl See proof of Lemma~\ref{lemma:niyogi geom result}.}
	  \label{fig:Proof picture}
\end{figure}
It is clear that $\angle [q-p,z-p]\le\pi/2$. Also since $\|z-q\| \ge \|z-p\|=\condition$, we have $\angle [z-q,p-q]\le \angle [z-p,q-p]\le\pi/2$. Therefore, $z'
$ is
indeed between $p$ and $q$. The angle between $l$ and the line passing through $p$ and $z$ equals the angle between $q-
    p$ and $\downarrow_p(q-p)$, that is
\begin{equation} \label{eq:1e1}
    \angle[p-z,z'-z]=\angle[q-p,\downarrow_p(q-p)]=:\theta_{p,q}.
\end{equation}
    To obtain an upper bound for $\theta_{p,q}$, we again note that $\|z-q\|\ge\|z-p\|$ and therefore  $\|z'-q\|\ge\|z'-p\|$, or $\|z'-p\|\le\frac{1}{2}\|q-p\|$. So, $\theta_{p,q}$ is bounded by $\sin^{-1}\left(\Vert q-p\Vert/2\condition\right)$. This completes the proof of Lemma \ref{lemma:niyogi geom result}.
\end{proof}

\begin{lemma}\label{Nonsing of projection}
    \emph{\cite[Lemma 5.4]{niyogi}}
    \notechecked{clarified statement of lemma} For $p\in\mathcal{M}$, the derivative of $\downarrow_{p}:\mathbb{R}^{N}\rightarrow\mathcal{T}_{p}$ is nonsingular on  $\mathcal{A}_{\mathcal{M}}(p,\tau/2)$.
\end{lemma}

\begin{lemma}\label{lemma:max curv}
    \emph{\cite[Proposition 6.1]{niyogi}}
    Let $\thresh(\cdot)$ denote a smooth unit-speed geodesic curve on $\manifold$ defined on an interval $I\subset\mathbb{R}$. For every $t\in I$, the following holds.
    \[
        \Vert\thresh''(t)\Vert\le 1/\condition.
    \]
\end{lemma}

\begin{lemma}\label{lemma:angle btw tangents}
    \emph{\cite[Proposition 6.2]{niyogi}}
    Fix $p,q\in \manifold$. The angle between $\tan_{p}$ and $\tan_{q}$, $\angle[\tan_p,\tan_q]$, satisfies $\cos(\angle[\tan_p,\tan_q])\ge 1-d_{\manifold}\left(p,q\right)/\condition$.
\end{lemma}
The next lemma guarantees that two points separated by a small
Euclidean distance are also separated by a small geodesic distance,
and so the manifold does not ``curve back'' upon itself.
\begin{lemma}\label{lemma:d & dM}
    \emph{\cite[Proposition 6.3]{niyogi}}
    For $p,q\in \manifold$ with $\|q-p\|\le\condition/2$,
    we have
	   \begin{equation}
       \label{eq:niyogiTrue}
            d_{\manifold}\left(p,q\right)\le\condition-\condition\sqrt{1-\frac{2}{\condition}\|q-p\|}.
	   \end{equation}
\end{lemma}
\begin{proof}
The first part of the proof of Proposition 6.3 in~\cite{niyogi} establishes that for any $p,q\in \manifold$,
\begin{equation}
\label{eq:niyogiCombined}
\|q-p\| \ge d_{\manifold}\left(p,q\right) - \frac{(d_{\manifold}\left(p,q\right))^2}{2\condition},
\end{equation}
which is satisfied only if \eqref{eq:niyogiTrue} is satisfied or if
\begin{equation}
\label{eq:niyogiFalse}
d_{\manifold}\left(p,q\right)\ge\condition+\condition\sqrt{1-\frac{2}{\condition}\|q-p\|}
\end{equation}
is satisfied. We provide the following argument to complete the proof.

For fixed $p \in \manifold$, let us consider
$$
\widehat{q} := \arg\min_{q \in \manifold, d_{\manifold}\left(p,q\right) \ge \condition} \|q-p\|.
$$
We know the minimizer $\widehat{q}$ exists because we are minimizing a continuous function over a compact set.  We consider two cases. First, if $d_{\manifold}\left(p,\widehat{q}\right) = \condition$, then by \eqref{eq:niyogiCombined}, we will have $\|\widehat{q}-p\| \ge \condition/2$. Second, if $d_{\manifold}\left(p,\widehat{q}\right) > \condition$, then there must exist an open neighborhood of $\widehat{q}$ on $\manifold$ over which the distance to $p$ is minimized at $\widehat{q}$. This implies that $p-\widehat{q}$ will be normal to $\manifold$ at $\widehat{q}$, which by the definition of condition number (and the fact that $p \in \manifold$) means that we must have $\|\widehat{q}-p\| \ge 2\condition$.

Now, for any $p,q \in \manifold$ such that $\|q-p\| < \condition/2$, \eqref{eq:niyogiTrue} would imply that $d_{\manifold}\left(p,q\right) < \condition$ and \eqref{eq:niyogiFalse} would imply that $d_{\manifold}\left(p,q\right) > \condition$. From the paragraph above, we see that if $d_{\manifold}\left(p,q\right) \ge \condition$, then $\|q-p\| \ge \|\widehat{q}-p\| \ge \condition/2$, and so we can rule out the possibility that \eqref{eq:niyogiFalse} is true. Thus, \eqref{eq:niyogiTrue} must hold for any $p,q \in \manifold$ with $\|q-p\| < \condition/2$.

For any $p,q \in \manifold$ such that $\|q-p\| = \condition/2$, \eqref{eq:niyogiTrue} would imply that $d_{\manifold}\left(p,q\right) \le \condition$ and \eqref{eq:niyogiFalse} would imply that $d_{\manifold}\left(p,q\right) \ge \condition$. From the paragraph above involving $\widehat{q}$, we see that any point $q \in \manifold$ satisfying both $d_{\manifold}\left(p,q\right) \ge \condition$ and $\|q-p\| = \condition/2$ would have to be a local minimizer of $\|q-p\|$ on the convex set and in fact would have to fall into the first case, implying that $d_{\manifold}\left(p,q\right) = \condition$ exactly. It follows that \eqref{eq:niyogiTrue} must hold for any $p,q \in \manifold$ with $\|q-p\| = \condition/2$.
\end{proof}

The next lemma concerns the invertibility of $\downarrow_p$ within the neighborhood of $p$ and is closely related to Lemma 5.3 in~\cite{niyogi}.
\begin{lemma}
	\label{lemma:injective map from Niyogi}
    For $p\in \manifold$, $\downarrow_p$ is invertible on $\intersection_\manifold(p,\condition/4)$.
\end{lemma}
\begin{proof}
	Lemma \ref{Nonsing of projection} states that the derivative of $\downarrow_{p}$ is nonsingular on $\intersection_\manifold(p,\condition/2)$. The inverse function theorem then
	implies that there exists an $r>0$ such that $\downarrow_{p}$ is invertible on $\intersection_{\manifold}(p,r\condition)$; without loss of generality assume that $r<1/4$ (otherwise we are done). Now, suppose that there exists $c>0$ and distinct points  $q,z\in\manifold$ such that $\|q-p\|=c\condition$, $\|z-p\|\le c\condition$, and $\downarrow_p(q-p)=\downarrow_p(z-p)$.
	In particular, this implies that
	\begin{equation}
		z-q\perp \tan_{p}.\label{eq:implication of noninvertibility}
	\end{equation}
	That is, for any unit vector $v\in \tan_{p}$, we have
	\begin{equation}
		\left\langle z-q,v\right\rangle =0.\label{eq:implication of noninvertibility 2}
	\end{equation}
	Our goal is to show that $c>1/4$. Suppose, in contradiction that indeed $c\le1/4$.
	Let $\curve(\cdot)$ be the unit-speed geodesic curve connecting $q$
	to $z$, such that $\curve(0)=q$ and $\curve(d_{\manifold}(q,z))=z$.
	By applying the fundamental theorem of calculus twice, we realize
	that
	\begin{align*}
		z-q & =\curve\left(d_{\manifold}(q,z)\right)-\curve(0)\\
		& =\int_{0}^{d_{\manifold}(q,z)}\curve'(\alpha)\, d\alpha\\
		& =\int_{0}^{d_{\manifold}(q,z)}\left(\curve'(0)+\int_{0}^{\alpha}\curve''(\beta)\, d\beta\right)\, d\alpha\\
		& =\curve'(0)\cdot d_{\manifold}(q,z)+\int_{0}^{d_{\manifold}(q,z)}\int_{0}^{\alpha}\curve''(\beta)\, d\beta d\alpha.
	\end{align*}
	Invoking Lemma \ref{lemma:max curv}, we can write that
	\begin{align}
		\Vert(z-q)-\curve'(0)\cdot d_{\manifold}(q,z)\Vert
		& \le\int_{0}^{d_{\manifold}(q,z)}\int_{0}^{\alpha}\Vert \curve''(\beta)\Vert\, d\beta d\alpha\nonumber \\
		& \le\frac{1}{\condition}\int_{0}^{d_{\manifold}(q,z)}\int_{0}^{\alpha}d\beta d\alpha\nonumber \\
		& =\frac{d_{\manifold}^{2}(q,z)}{2\condition}.\label{eq:diff btw tangent and connecting segment}
	\end{align}
	Meanwhile, having $\Vert z-q\Vert\le2c\condition$ implies, via Lemma \ref{lemma:d & dM}, that
	\begin{equation}
		 d_{\manifold}(q,z)\le\condition-\condition\sqrt{1-4c},\label{eq:geodesic dist of q & z}
	\end{equation}
	which, after plugging back into \eqref{eq:diff btw tangent and connecting segment},
	yields
	\begin{equation}
	\left\Vert \frac{z-q}{d_{\manifold}(q,z)}-\curve'(0)\right\Vert \le\frac{1}{2}-\frac{1}{2}\sqrt{1-4c}.\label{eq:diff btw tan and connecting line 2}
	\end{equation}
	So, for any unit vector $v\in \tan_{p}$, we have
	\begin{align}
		\left|\left\langle \curve'(0),v\right\rangle \right| & \le\left|\left\langle \curve'(0)-\frac{z-q}{d_{M}(q,z)},v\right\rangle \right|+\left|\left\langle \frac{z-q}{d_{M}(q,z)},v\right\rangle \right|\nonumber \\
		& =\left|\left\langle \curve'(0)-\frac{z-q}{d_{M}(q,z)},v\right\rangle \right|\nonumber \\
		& \le\left\Vert \curve'(0)-\frac{z-q}{d_{M}(q,z)}\right\Vert \nonumber \\
		& \le\frac{1}{2}-\frac{1}{2}\sqrt{1-4c},\label{eq:contradiction part 1}
	\end{align}
	where the first line follows from the triangle inequality, and the
	second line uses \eqref{eq:implication of noninvertibility 2}. The
	last line uses \eqref{eq:diff btw tan and connecting line 2}. To
	reiterate, \eqref{eq:contradiction part 1} is valid for any unit vector
	$v\in \tan_{p}$.

	On the other hand, Lemma \ref{lemma:angle btw tangents} implies that
	\begin{align}
		\cos\left(\angle [\tan_{q},\tan_{p}]\right) & \ge1-\frac{1}{\condition}d_{\manifold}(p,q)\nonumber \\
		& \ge \sqrt{1-\frac{2}{\condition}\Vert q-p\Vert}\nonumber\\
		& =\sqrt{1-2c},\label{eq:bnd on the ang btw subspaces}
	\end{align}
	where the second line follows from Lemma \ref{lemma:d & dM}, and the last line uses $\Vert q-p\Vert=c\condition$. By the definition of the angle between
	subspaces, \eqref{eq:bnd on the ang btw subspaces} implies that there
	exists a unit vector $v_{0}\in \tan_{p}$ such that
	\begin{equation}
		\left|\left\langle v_{0},\curve'(0)\right\rangle \right|\ge\sqrt{1-2c}\label{eq:contradiction part 2}
	\end{equation}
	because $\curve'(0)\in\tan_q$. Combining this bound with \eqref{eq:contradiction part 1} for $v=v_{0}$,
	we realize that
	\[
		  \sqrt{1-2c}\le\frac{1}{2}-\frac{1}{2}\sqrt{1-4c}.
	\]
	This inequality is not met for any $c\le1/4$. Thus, indeed $c>1/4$.
	In particular, this means that $\downarrow_{p}$ is invertible on $\intersection_\manifold(p,\condition/4)$.
\end{proof}

The next three lemmas are of importance when approximating the long and short chords on $\manifold$ with, respectively, nearby long chords and vectors on the nearby tangent planes.

\begin{lemma}\emph{\cite[Implicit in Lemma 4.1]{clarkson2008tighter}}\label{lemma:implicit in Clarkson}
	Consider two pair of points $a_{1},a_{2}$ and $b_{1},b_{2}$, all in $\real^N$, such that $\left\Vert a_{1}-b_{1}\right\Vert ,\left\Vert a_{2}-b_{2}\right\Vert \le r$, and that $\left\Vert a_{1}-a_{2}\right\Vert \ge \kappa\sqrt{r}$, for some $r,\kappa>0$. Then $\left\Vert U\left(a_{1},a_{2}\right)-U\left(b_{1},b_{2}\right)\right\Vert \le4\kappa^{-1}\sqrt{r}$.
\end{lemma}

\begin{lemma}\label{lemma:difference of projectors on nearby tangent planes}
	For $a,b\in\manifold$ with $\|a-b\|\le l_1<\condition/2$, it holds true that
	\begin{equation*}
		\|\downarrow_a v-\downarrow_b v\|\le \sqrt{\frac{2l_1}{\condition}},
	\end{equation*}
	for every unit vector $v\in\real^N$.
\end{lemma}
\begin{proof}
	It follows from \eqref{eq:principal angle and projection norms} that
	\begin{equation}\label{eq:op norm and angle btw subspaces}
		\|\downarrow_a v-\downarrow_b v\|\le\|(\downarrow_a-\downarrow_b)(\cdot)\|_{2,2}=\sin(\angle[\tan_a,\tan_b]).
	\end{equation}
On the other hand, since $\|a-b\|\le l_{1}<\condition/2$, Lemma \ref{lemma:d & dM} implies that
	\[
		d_{\manifold}\left(a,b\right)\le\condition-\condition\sqrt{1-\frac{2l_{1}}{\condition}},
	\]
	 and thus, using Lemma \ref{lemma:angle btw tangents}, we arrive at
	\[
		\cos\left(\angle\left[\tan_{a},\tan_{b}\right]\right)\ge\sqrt{1-\frac{2l_{1}}{\condition}}.
	\]
	Plugging back the estimate above into \eqref{eq:op norm and angle btw subspaces}, we conclude that $\|\downarrow_a v-\downarrow_b v\|\le \sqrt{2l_1/\condition}$, as claimed.
\end{proof}

\begin{lemma}\label{lemma:distance of a shorrt chord to a nearby tangent plane}
	Fix $p\in\manifold$, and take two points $x_1,x_2\in\manifold$ such that $\|x_1-p\|\le l_1$ and $\|x_2-x_1\|\le l_2$, $l_1,l_2<\condition/2$. Then, we have that
	\[
		\|U(x_1,x_2)-\downarrow_p U(x_1,x_2)\|\le \sqrt{\frac{2l_1}{\condition}}+\frac{l_2}{2\condition}.
	\]
\end{lemma}
\begin{proof}
	The triangle inequality implies that
	\begin{equation}\label{eq:triangle ineq for the distance btw short chord and tangent plane}
		\|U(x_1,x_2)-\downarrow_p U(x_1,x_2)\|\le \|U(x_1,x_2)-\downarrow_{x_1} U(x_1,x_2)\|+\|\downarrow_{x_1}U(x_1,x_2)-\downarrow_p U(x_1,x_2)\|.
	\end{equation}
	Since $\|x_{2}-x_{1}\|\le l_{2}<2\condition$, Lemma \ref{lemma:niyogi geom result}
	is the right tool to bound the first term on the right hand side of \eqref{eq:triangle ineq for the distance btw short chord and tangent plane}:
	\begin{align}
		\left\Vert U\left(x_{1},x_{2}\right)-\downarrow_{x_{1}}U\left(x_{1},x_{2}\right)\right\Vert  & =\sin\left(\angle\left[U\left(x_{1},x_{2}\right),\downarrow_{x_{1}}U\left(x_{1},x_{2}\right)\right]\right)\nonumber \\
		& =\sin\left(\angle\left[x_{2}-x_{1},\downarrow_{x_{1}}(x_{2}-x_{1})\right]\right)\nonumber \\
		& \le\frac{l_{2}}{2\condition}.\label{eq:triangle p1}
	\end{align}
	Since $\|x_1-p\|\le l_1<\condition/2$, a bound on the second term follows from an application of Lemma \ref{lemma:difference of projectors on nearby tangent planes}:
	\begin{equation}
		\|\downarrow_{x_1}U(x_1,x_2)-\downarrow_p U(x_1,x_2)\|\le \sqrt{\frac{2l_1}{\condition}}.\label{eq:triangle p2}
	\end{equation}
	Combining \eqref{eq:triangle p1} and \eqref{eq:triangle p2} immediately proves our claim.
\end{proof}

\notechecked{moved the next lemma (and its proof) here from later appendices} We will also need the following result regarding the local properties of $\manifold$, which is closely related to Lemma 5.3 in~\cite{niyogi}.
	
\begin{lemma}\label{lemma:revised Niyogi 5.3}
Let $p\in\manifold$ and $r\le\condition/4$. Then the following holds:
	\begin{equation*}
		\emph{\mbox{vol}}_K(\intersection_\manifold(p,r))\ge \left(1-\frac{r^2}{4\condition^2}\right)^{\frac{K}{2}}\,r^K V_{\ball_K},
	\end{equation*}
	where {\em $\Vol_K(\cdot)$} measures the $K$-dimensional volume.
\end{lemma}
\begin{proof}
	As in the proof of Lemma 5.3 in \cite{niyogi}, we will show that for some $r'>0$ to be defined below,
	\[
		B_{\tan_p}(r')\subset\downarrow_p(A_\manifold(p,r)),
	\]
	as our claim follows directly from the inclusion above. To show the above inclusion, we use the following argument. Let us denote the inverse of $\downarrow_p$ on $\intersection_\manifold(p,\condition/4)$ with $g(\cdot)$.

	From Lemma~\ref{lemma:injective map from Niyogi}, $\downarrow_p$ is invertible on $A_\manifold(p,r)$ and therefore  $\downarrow_p(\intersection_\manifold(p,r))$ is an open set. Thus there exists $s>0$ such that $B_{\tan_p}(s)\subset\downarrow_p(\intersection_\manifold(p,r))$. We can keep increasing $s$ until at $s=s^*$
	we reach a point $y$  on the boundary of the closure of $B_{\tan_p}(s^\ast)$ such that $y\notin\downarrow_p(\intersection_\manifold(p,r))$. Consider a sequence $\{y_i\}\subset B_{\tan_p}(s^*)\subset\downarrow_p(\intersection_\manifold(p,r))$ such that $y_i\rightarrow y$ when $i\rightarrow\infty$. Note that $\{g(y_i)\}\subset \intersection_\manifold(p,r)$ and, because every sequence in a compact space contains a convergent subsequence, there exist a convergent sebsequence $\{g(y_{i_k})\}$ and $x$ in the closure of $\intersection_\manifold(p,r)$ such that $g(y_{i_k})\rightarrow x$.
	Since $\downarrow_p$ is continuous, $\downarrow_p x=y$. Therefore  $y=\downarrow_p x\notin\downarrow_p(\intersection_\manifold(p,r))$, and $x\notin \intersection_\manifold(p,r)$ and thus $x$ is on the boundary of the closure of $\intersection_\manifold(p,r)$ and $\|x\|=r$. Now we invoke Lemma \ref{lemma:niyogi geom result} with $q=x+p$ to obtain that
	\[
		\cos(\angle[x,y])\ge \sqrt{1-\frac{r^2}{4\condition^2}}.
	\]
	It follows that
	\begin{align*}
		s^*&=\Vert y\Vert\\
		& =\cos(\angle[x,y])\cdot r\\
		&\ge \sqrt{1-\frac{r^2}{4\condition^2}}\cdot r
		=:r',
	\end{align*}
	and thus $B_{\tan_p}(r')\subset \downarrow_{p}(\intersection_{\manifold}(p,r))$.
This completes the proof of Lemma \ref{lemma:revised Niyogi 5.3} since
$$
\Vol_K(\intersection_{\manifold}(p,r))\ge \Vol_K(\downarrow_p(\intersection_{\manifold}(p,r)))\ge \Vol_K(\ball_{\tan_p}(r'))=(r')^K\Vol(\ball_K),
$$
where the first inequality holds because projection onto a subspace is non-expansive.
\end{proof}

\notechecked{the lemma below and its proof are new} We close this section with a list of  properties of the Dirichlet kernel which are later used in the proof of Lemma \ref{lemma:example} (about the condition number of the complex exponential curve).
\begin{lemma}\label{lem:(Dirichlet-kernel)}
(Dirichlet kernel)\label{lem:(Properties of Dirichlet-kernel)} For $z\in[-1/2,1/2]$, the Dirichlet kernel takes $z$ to
\[
D_{N}(z):=\frac{\sin(\pi Nz)}{\sin(\pi z)}.
\]
If $|z|> 2/N$, then it holds that
\[
\left|D_{N}(z)\right|\le\alpha_{1}N,
\]
with $\alpha_1\approx 0.23$.
Moreover, there exists some $\alpha_2>0$ and $N_{2}:=N_2(\alpha_2)$, such that the following holds for every $N>N_2$:
$$
|D_N(z)|\le N\left(1-\frac{(N\pi z)^2}{40}\right) + \alpha_2 N z^2
$$
for all $|z|\le 2/N$.
\end{lemma}

\begin{proof}
According to~\cite[Table 7.2]{oppenheim2011discrete}, the relative peak side-lobe amplitude of the Dirichlet kernel is (approximately) $-13$ decibels. That is,  the peak side-lobe of the Dirichlet kernel is no larger than $\alpha_1N$ with $\alpha_1\approx 0.23$.  It is also easily verified that this peak does not occur further than $2/N$ away from the origin. To summarize,
$$
|D_N(z)|\le \alpha_1 N,
$$
as long as $|z|>2/N$. This completes the proof of the first inequality in Lemma \ref{lem:(Dirichlet-kernel)}. To prove the second inequality, assume that $|z|\le 2/N$. As $N\rightarrow\infty$, any $z\in [-2/N,2/N]$ approaches zero and we may replace the sine in the denominator of the Dirichlet kernel with its argument. That is, as $N\rightarrow\infty$, $z\rightarrow 0$ and
\begin{align*}
\frac{\left|D_{N}(z)\right|}{N} & =\frac{|\sin( N\pi z)|}{N|\sin(\pi z)|}\\
 & =\frac{|\sin(N\pi z)|}{N\pi |z|(1+O(z^2))}\\
  & \le \frac{|\sin(N\pi z)|}{N\pi |z|}(1+O(z^2))\\
  & \le \frac{ |(N\pi z)-\frac{1}{40}(N\pi z)^3|}{N\pi |z|}(1+O(z^2))\\
    & = \left(1-\frac{(N\pi z)^2}{40}\right)(1+O(z^2)),
\end{align*}
where the third line uses the fact that $1/(1-a)\le 1+2a$ for all $0\le a\le 1/2$. The second to last line holds because $\sin a\le a-a^3/40$ for all $ 0\le a\le 2\pi$. As a result, for  some $\alpha_2>0$ and $N_2=N_2(\alpha_2)$, the following holds for every $N>N_2$:
\begin{equation}
\frac{\left|D_{N}(z)\right|}{N} \le  \left(1-\frac{(N\pi z)^2}{40}\right)(1+O(z^2))
\le \left(1-\frac{(N\pi z)^2}{40}\right) +\alpha_2z^2,
\end{equation}
which, to reiterate, holds as long as $|z|\le 2/N$. This completes the proof of Lemma \ref{lem:(Dirichlet-kernel)}.
\end{proof}

\section{Proof of Lemma~\ref{lemma:example}}\label{sec:example}

\notechecked{this lemma and its proof are new} Here, $\tau_{\beta}$ stands for  the reach of the complex exponential curve, i.e., the inverse of its condition number. Note that the reach of the complex exponential curve is defined as the largest $d\ge0$ such that every point within an $\ell_{2}$ distance less than $d$ from $\beta$ has a unique nearest point (in the $\ell_{2}$ sense) on $\beta$. In the rest of the proof, ($i$) we first find a unit-speed parametrization of $\beta$, ($ii$) we then derive some basic properties of the reparametrized curve, and ($iii$) finally, we estimate $\tau_{\beta}$ by studying the long and short chords on the reparametrized curve separately.

\subsection{Unit-speed geodesic on $\beta$}

Let $\gamma:\mathbb{R}\rightarrow\mathbb{C}^{N}$ be a unit-speed geodesic obtained by appropriately normalizing $\beta$. For every $s\in\mathbb{R}$, there must exist $t=t(s)\in\mathbb{R}$ such that $\gamma_{s}=\beta_{t}$. In particular, we note that $\beta$ is a constant-speed curve with
\[
\left\Vert \frac{d\beta_{t}}{dt}\right\Vert = \left(\sum_{n=-f_{C}}^{f_{C}}\left(2\pi n\right)^{2}\right)^{1/2}=\frac{2\pi}{\sqrt{3}}\left(f_{C}\left(f_{C}+1\right)\left(2f_{C}+1\right)\right)^{1/2}=:\frac{1}{v_{N}}=O(f_C^{3/2}),
\]
and therefore we can simply take $t(s)=v_{N}s$. This gives
\begin{equation}
\gamma_{s}=\beta_{t(s)}=\beta_{v_{N}s}=\left[\begin{array}{c}
e^{-\uniti2\pi f_{C}v_{N}s}\\
e^{-\uniti2\pi(f_{C}-1)v_{N}s}\\
\vdots\\
e^{\uniti2\pi(f_{C}-1)v_{N}s}\\
e^{\uniti2\pi f_{C}v_{N}s}
\end{array}\right],\label{eq:gamma}
\end{equation}
\begin{equation}
\frac{d\gamma_{s}}{ds}=v_{N}\left[\begin{array}{c}
-\uniti2\pi f_{C}\cdot e^{-\uniti2\pi f_{C}v_{N}s}\\
-\uniti2\pi(f_{C}-1)\cdot e^{-\uniti2\pi(f_{C}-1)v_{N}s}\\
\vdots\\
\uniti2\pi(f_{C}-1)\cdot e^{\uniti2\pi(f_{C}-1)v_{N}s}\\
\uniti2\pi f_{C}\cdot e^{\uniti2\pi f_{C}v_{N}s}
\end{array}\right], ~\text{and}\label{eq:gammap}
\end{equation}
\begin{equation}
\frac{d^{2}\gamma_{s}}{ds^{2}}=-v_{N}^{2}\left[\begin{array}{c}
(2\pi f_{C})^{2}\cdot e^{-\uniti2\pi f_{C}v_{N}s}\\
(2\pi(f_{C}-1))^{2}\cdot e^{-\uniti2\pi(f_{C}-1)v_{N}s}\\
\vdots\\
(2\pi(f_{C}-1))^{2}\cdot e^{\uniti2\pi(f_{C}-1)v_{N}s}\\
(2\pi f_{C})^{2}\cdot e^{\uniti2\pi f_{C}v_{N}s}
\end{array}\right].\label{eq:gammapp}
\end{equation}
To reiterate, \eqref{eq:gamma} and \eqref{eq:gammap} represent $\gamma$ (a unit-speed parametrization of $\beta$) and its tangent vector. In addition, the curvature at any point can be computed as the magnitude of the second derivative in \eqref{eq:gammapp}. That is,
\begin{equation}
\left\Vert \frac{d^{2}\gamma_s}{ds^{2}}\right\Vert  =v_{N}^{2}\left(\sum_{-f_{C}}^{f_{C}}\left(2\pi n\right)^{4}\right)^{1/2} =\left(\sum_{-f_{C}}^{f_{C}}\left(2\pi n\right)^{2}\right)^{-1}\left(\sum_{-f_{C}}^{f_{C}}\left(2\pi n\right)^{4}\right)^{1/2}=:w_{N}=O(f_C^{-1/2}),\label{eq:curv}
\end{equation}
where we used (\ref{eq:gammapp}). Observe that the curvature is constant and scales like $1/\sqrt{N}$ for large $N$.

Since $\gamma$ is periodic, we will use $t_{1}\ominus t_{2}$ to denote subtraction modulo $1$ for any $t_1,t_2\in \mathbb{R}$ so that
$$
t_1- t_2 = \lfloor t_1-t_2\rfloor + (t_1\ominus t_2).
$$
(Equivalently, $\ominus$ represents the natural subtraction on the unit circle.) We continue by recording a few simple facts about the reparametrized complex exponential curve $\gamma$.

\subsection{Some observations about $\gamma$}

Note that $\gamma_{s}$ as a zero-padded sequence in $\ell_{2}(\mathbb{Z})$ can be interpreted as the (reversed) sequence of Fourier series coefficients of the signal in time that, at $t\in\mathbb{R}$, takes the value
\begin{align*}
\check{\gamma}_{s}(t) & =\frac{\sin(\pi N\left(t\ominus v_{N}s\right))}{\sin(\pi(t\ominus v_{N}s))}=:D_{N}(t\ominus v_{N}s),
\end{align*}
where $D_{N}(\cdot)$ is the Dirichlet kernel of width $\sim2/N$. The Dirichlet kernel is known to decay rapidly outside of an interval of width $\sim2/N$ centered at the origin as studied in Lemma~\ref{lem:(Dirichlet-kernel)} in the Toolbox.
One immediate consequence of Lemma \ref{lem:(Dirichlet-kernel)} is that
\begin{equation}
\left|\left\langle D_{N}(\cdot\ominus t_{1}),D_{N}(\cdot\ominus t_{2})\right\rangle \right|=\left|D_{N}(t_{1}\ominus t_{2})\right|\le\alpha_{1}N \qquad\mbox{if } t_{1} \ominus t_{2} \in [2/N,1-2/N].\label{eq:useful}
\end{equation}
The first identity above holds because circular convolution of the Dirichlet kernel with itself produces the Dirichlet kernel again. Now, for any pair $s_{1},s_{2}\in\mathbb{R}$, consider the following correlation:
\begin{align}
\left|\left\langle \gamma_{s_{1}},\gamma_{s_{2}}\right\rangle \right|=\left|\left\langle \check{\gamma}_{s_{1}},\check{\gamma}_{s_{2}}\right\rangle \right| = \left|\left\langle D_{N}(\cdot\ominus v_{N}s_{1}),D_{N}(\cdot\ominus v_{N}s_{2})\right\rangle \right|
=\left|D_{N}(v_{N}s_{1}\ominus v_{N}s_{2})\right|,\label{eq:-1}
\end{align}
where we used the Plancherel identity above. Then it follows from  (\ref{eq:useful}) that
\begin{equation}\label{eq:0}
\left|\left\langle \gamma_{s_{1}},\gamma_{s_{2}}\right\rangle \right|\le\alpha_{1}N \qquad\mbox{if } v_{N}s_{1}\ominus v_{N}s_{2}\in [2/N,1-2/N].
\end{equation}
In words, \eqref{eq:0} captures the long-distance correlations on $\gamma$. We now turn our attention to short-distance correlations. According to Lemma \ref{lem:(Dirichlet-kernel)}, for some $\alpha_2>0$ and  $N_2=N_2(\alpha_2)$, the following holds for every $N>N_2$:
\begin{equation}
\left|\left\langle \gamma_{s_{1}},\gamma_{s_{2}}\right\rangle \right|
\le N\left(1-\frac{N^2\pi^2}{40}\left(v_{N}s_{1}\ominus v_{N}s_{2}\right)^{2}\right)+\alpha_2N\left(v_{N}s_{1}\ominus v_{N}s_{2}\right)^{2} \label{eq:3}
\end{equation}
if $v_{N}s_{1}\ominus v_{N}s_{2}\in[0,2/N]$, where we used \eqref{eq:-1} again. If $v_{N}s_{1}\ominus v_{N}s_{2} \in [1-2/N,1]$, then \eqref{eq:3} holds with $v_{N}s_{1}\ominus v_{N}s_{2}$ replaced by $1 - (v_{N}s_{1}\ominus v_{N}s_{2}) = v_{N}s_{2}\ominus v_{N}s_{1}$. The conclusion in \eqref{eq:3} is a direct consequence of the vanishing derivative of the Dirichlet kernel at the origin. We are now in a position to estimate the reach of the complex exponential curve.

\subsection{Estimating $\tau_{\beta}$}

Consider a point $\gamma_{s}\in\mathbb{\mathbb{C}}^{N}$ on the
complex exponential curve for an arbitrary $s\in\mathbb{R}$. We deviate from $\gamma_{s}$
by $\chi$ to obtain $x=\gamma_{s}+\chi$ where $\chi$ is assumed to be normal
to the complex exponential curve at $\gamma_{s}$, that is
\begin{equation}
\left\langle \chi,\frac{d\gamma_{s}}{ds}\right\rangle =0,\label{eq:norm cond}
\end{equation}
where $d\gamma_{s}/ds$ is the tangent vector at $\gamma_{s}$ (which was computed in (\ref{eq:gammap})).

We seek the largest $d> 0$ such that for all $\chi$ with $\|\chi\|<d$ and
satisfying (\ref{eq:norm cond}), $\gamma_{s}$ is the unique nearest
point to $x=\gamma_{s}+\chi$ on the complex exponential curve. For $\gamma_{s}$ to
be the unique nearest point to $x$, it must hold that
\begin{equation}
\|\chi\|=\|x-\gamma_{s}\|<\|x-\gamma_{s'}\|=\|\chi+\gamma_{s}-\gamma_{s'}\|,\qquad\forall v_{N}s'\ominus v_{N}s\ne0.\label{eq:e1 1}
\end{equation}
Now (\ref{eq:e1 1}) is equivalent to
\begin{equation}
\mbox{Re} [\left\langle \chi,\gamma_{s'}-\gamma_{s}\right\rangle] <N-\mbox{Re}[\left\langle \gamma_{s},\gamma_{s'}\right\rangle]
= N- \left\langle \gamma_{s},\gamma_{s'}\right\rangle,\qquad\forall v_{N}s'\ominus v_{N}s\ne0,\label{eq:1}
\end{equation}
where we used the fact that the complex exponential curve lives on a sphere of radius
$\sqrt{N}$ in $\mathbb{C}^{N}$. We consider two separate cases:

\paragraph{Long distances $v_{N}s'\ominus v_{N}s \in [2/N,1-2/N]$: }

In this case, it follows from (\ref{eq:0}) that
\[
\left|\left\langle \gamma_{s},\gamma_{s'}\right\rangle \right|\le\alpha_{1}N,
\]
with $\alpha\approx 0.23$ as in Lemma~\ref{lem:(Dirichlet-kernel)}.
As a result, for the inequality in (\ref{eq:1}) to hold for long distances, it suffices that the following
holds:
\[
\|\chi\|\cdot\left(\|\gamma_{s}\|+\|\gamma_{s'}\|\right)=\|\chi\|\cdot2\sqrt{N}<\left(1-\alpha_{1}\right)N,\qquad\forall v_{N}s'\ominus v_{N}s\in[2/N,1-2/N],
\]
which is guaranteed as long as
\begin{equation}
\|\chi\|<\frac{\left(1-\alpha_{1}\right)N}{2\sqrt{N}}=\frac{1-\alpha_{1}}{2}\sqrt{N}.\label{eq:short guarantee}
\end{equation}

\paragraph{Short distances $v_{N}s'\ominus v_{N}s\in [0,2/N]\cup[1-2/N,1]$: }

Without loss of generality assume that $v_{N}s'- v_{N}s\in[0,2/N]$.
In this case, we first note that
\begin{align*}
\left|\left\langle \chi,\gamma_{s'}-\gamma_{s}\right\rangle \right|
& =\left|\left\langle \chi,\int_{s}^{s'}\frac{d\gamma_{\eta}}{ds}\, d\eta\right\rangle \right|\\
& =\left|\left\langle \chi,
\int_{s}^{s'} \frac{d\gamma_{s}}{ds} \, d\eta+ \int_{s}^{s'}\int_{s}^{\eta}\frac{d^{2}\gamma_{\xi}}{ds^{2}}\, d\xi d\eta
\right\rangle \right|\\
& =\left|\left\langle \chi,
(s'-s)\frac{d\gamma_{s}}{ds} + \int_{s}^{s'}\int_{s}^{\eta}\frac{d^{2}\gamma_{\xi}}{ds^{2}}\, d\xi d\eta
\right\rangle \right|\\
& =\left|\left\langle \chi,
\int_{s}^{s'}\int_{s}^{\eta}\frac{d^{2}\gamma_{\xi}}{ds^{2}}\, d\xi d\eta
\right\rangle \right|\\
 & \le\|\chi\|\cdot\int_{s}^{s'}\int_{s}^{\eta}\left\Vert \frac{d^{2}\gamma_{\xi}}{ds^{2}}\right\Vert \, d\xi d\eta\\
 & =w_{N}\|\chi\|\int_{s}^{s'}\int_{s}^{\eta}\, d\xi d\eta\\
 & = w_{N}\|\chi\|\cdot\frac{\left|s'-s\right|^{2}}{2}\\
  & = w_{N}\|\chi\|\cdot \frac{\left(v_N s'\ominus v_N s\right)^{2}}{2 v_N^2},
\end{align*}
where we used the fundamental theorem of calculus twice. The fourth line above uses
 the fact that $\chi$ is normal to the tangent of $\gamma$ at $s$,
namely $\left\langle \chi,d\gamma_s/ds\right\rangle =0$. The sixth line
uses the fact that curvature of $\gamma$ is constant and was calculated
in (\ref{eq:curv}).

Recall that  $v_{N}s'\ominus v_{N}s\le2/N$. Therefore, for a fixed $\alpha_2>0$ and $N>N_2=N_2(\alpha_2)$,  (\ref{eq:3}) dictates that
\[
\left\langle \gamma_{s'},\gamma_{s}\right\rangle \le N\left(1-\frac{N^2\pi^2}{40}\left(v_{N}s'\ominus v_{N}s\right)^{2}\right)+\alpha_2N\left(v_{N}s'\ominus v_{N}s\right)^{2}.
\]
As a result, for the inequality (\ref{eq:1}) to hold for short distances, it suffices that the following
statement holds:
\[
w_{N}\|\chi\|\cdot\frac{\left(v_N s'\ominus v_{N}s\right)^{2}}{2v_N^2}<\frac{N^3\pi^{2}}{40}\left(v_{N}s'\ominus v_{N}s\right)^{2}-\alpha_2N\left(v_{N}s'\ominus v_{N}s\right)^{2},\qquad\forall v_N s'- v_N s\in[0,2/N],
\]
which in turn holds if
\begin{equation}
\|\chi\|<\frac{\pi^2}{20}\cdot\frac{N^3v_N^2}{w_{N}}-
 2\alpha_2\cdot\frac{Nv_N^2}{w_N}.\label{eq:long guarantee}
\end{equation}

\paragraph{A lower bound on the reach:}

From \eqref{eq:long guarantee} and \eqref{eq:short guarantee}, we overall observe that if
\[
\|\chi\|<\mathrm{min}\left(\frac{1-\alpha_{1}}{2}, ~ \left(\frac{\pi^2}{20}\cdot\frac{N^{5/2}v_N^2}{w_{N}}-
 2\alpha_2\cdot\frac{\sqrt{N}v_N^2}{w_N}\right)\right)\sqrt{N}=O\left(\sqrt{N}\right),
\]
then \eqref{eq:1} holds uniformly regardless of the value of $s$. Therefore, we find the following
lower bound on the reach of the complex exponential curve:
\[
\tau_{\beta}\ge\mathrm{min}\left(\frac{1-\alpha_{1}}{2}, ~ \left(\frac{\pi^2}{20}\cdot\frac{N^{5/2}v_N^2}{w_{N}}-
 2\alpha_2\cdot\frac{\sqrt{N}v_N^2}{w_N}\right)\right)\sqrt{N},
\]
which, to reiterate, holds for some $\alpha_2$ and every $N>N_2$. Because the factor multiplying $\alpha_2$ scales with $N^{-2}$ whereas the factor multiplying $\pi^2/20$ scales with $1$, the following holds for every $N>N_{\text{sine}}$ for some $N_{\text{sine}}$:
\begin{equation}
\tau_{\beta}\ge\mathrm{min}\left(\frac{1-\alpha_{1}}{2}, ~ \frac{\pi^2}{40}\cdot\frac{N^{5/2}v_N^2}{w_{N}}\right)\sqrt{N}=O\left(\sqrt{N}\right).\label{eq:lower bound on reach}
\end{equation}

\paragraph{An upper bound on the reach:}

Since the complex exponential curve lives on the unit sphere in $\mathbb{C}^{N}$, $\gamma_{s}$
is normal to $\gamma$ at any arbitrary $s$. This immediately implies
the following upper bound on the reach:
\begin{equation}
\tau_{\beta}\le\sqrt{N}.\label{eq:upper bound on reach}
\end{equation}
Together, \eqref{eq:lower bound on  reach} and \eqref{eq:upper bound on reach} complete the proof of Lemma \ref{lemma:example}.

\section{Proof of Theorem \ref{thm:new manifold embedding} }
\label{sec:Proof of new manifold embedding}

It is easily verified that our objective is to find an upper bound for
\begin{equation*}
	\pr{\sup_{y\in U(\manifold)}|\|\proj y\|-1|>\epsilon},
\end{equation*}
when $\epsilon\le1/3$.

The remainder of this section is divided to two parts. In the first part, we construct a sequence of increasingly finer nets for $\manifold$. This is in turn used to construct a sequence of covers for the set of all (normalized) secants in \notechecked{changed ``directions'' to ``(normalized) secants''} $\manifold$. In the second part, we apply a chaining argument that utilizes this later sequence of covers to prove Theorem~\ref{thm:new manifold embedding}.

\subsection{Sequence of covers for $U(\manifold)$}
\label{sec:Constructing a Sequence of Covers for U(M)}

For $\Mresolution>0$, let $C_0(\Mresolution)\subset\manifold$ denote a minimal $\Mresolution$-net for $\manifold$ over all $\Mresolution$-nets that are a subset of $\manifold$. Upper and lower bounds for $\#C_0(\Mresolution)=\coverno_\manifold(\Mresolution)$ are known for sufficiently small $\Mresolution$~\cite{niyogi}, where $\#C_0(\Mresolution)$ denotes the cardinality of $C_0(\Mresolution)$. Since the claim below slightly differs from the one in~\cite{niyogi}, the proof is included here.

\begin{lemma}
	\label{lemma:niyogi covering no lemma}
	When  $\Mresolution\le\condition/2$, it holds that
	\begin{equation}
		\#C_{0}\left(\Mresolution\right)\le
		 \left(\frac{2}{\theta\left(\Mresolution/4\condition\right)\Mresolution}\right)^{K}\frac{\volume}{V_{\ball_K}}=:\coverbnd_0(\Mresolution),
		\label{eq:niyogi covering number 1}
	\end{equation}
	where $\theta\left(\alpha\right):=\sqrt{1-\alpha^{2}}$ for $|\alpha|\le1$.
\end{lemma}
\notechecked{proof moved to here}
\begin{proof}
Using \eqref{eq:coverno n packing} and a simple volume comparison argument, we observe that
\[
	\#C_0(\Mresolution)=\coverno_\manifold(\Mresolution)\le\packing_\manifold(\Mresolution)\le\frac{\volume}{\inf_{p\in\manifold}\vol{K}{A_\manifold(p,\Mresolution/2)}}.
\]
Since $\Mresolution/2\le\condition/4$, we can apply Lemma \ref{lemma:revised Niyogi 5.3} from the Toolbox (with $r=\Mresolution/2$) and obtain that
\begin{align*}
	\#C_0(\Mresolution)& \le\frac{\volume}{\left(1-\frac{\Mresolution^2}{16\condition^2}\right)^{\frac{K}{2}}\left(\frac{\Mresolution}{2}\right)^KV_{\ball_K}}\\		& =\left(\frac{2}{\theta(\Mresolution/4\condition)\Mresolution}\right)^K\frac{\volume}{V_{\ball_K}},
\end{align*}
This completes the proof of Lemma \ref{lemma:niyogi covering no lemma}.
\end{proof}

By replacing $\Mresolution$ with $4^{-j}\Mresolution$, we can construct a sequence of increasingly finer nets for $\manifold$, $\{C_j(\Mresolution)\}$, such that $C_j(\Mresolution)\subset\manifold$ is a ($4^{-j}\Mresolution$)-net for $\manifold$, for every $j\in\ints$.  In light of Lemma \ref{lemma:niyogi covering no lemma}, we have that
\begin{equation}\label{eq:from Cj to C0}
	\#C_{j}(\Mresolution)\le 4^{jK}\cdot\coverbnd_0(\Mresolution).
\end{equation}

Construction of a sequence of covers for $U(\manifold)$ demands the following setup. For $\Tresolution>0$ and $j\in\ints$, let $C^{j}(\Tresolution)$ denote a minimal ($2^{-j}\Tresolution$)-net for $\ball_K$.
For $p\in C_j(\Mresolution)$, we can naturally map $C^j(\Tresolution)$ to live in the $K$-dimensional unit ball along $\tan_p$ (and anchored at the origin). We represent this set of vectors by $C^{j,p}(\Tresolution)$ and define
\[
	C'_j(\Mresolution,\Tresolution):=\bigcup_{p\in C_j(\Mresolution)} C^{j,p}(\Tresolution),
\]	
which forms a ($2^{-j}\Tresolution$)-net for the unit balls along the tangent spaces at every point in  $C_j(\Mresolution)\subset\manifold$. For $\Uresolution>0$, let us specify $\Mresolution$ and $\Tresolution$ as functions of $\Uresolution$.
For $\Cl{eta},\Cl{eta'}>0$ to be set later, take $\Mresolution=\Mresolution(\Uresolution)=\Cr{eta}^2\condition\Uresolution^2$ and $\Tresolution=\Tresolution(\Uresolution)=\Cr{eta'}\sqrt{\Mresolution/\condition}=\Cr{eta}\Cr{eta'}\Uresolution$.
Now, for every $j\in\ints$, simply set
\begin{equation*}
	\Ucover_j(\Uresolution):=U(C_j(\Mresolution))\cup C'_j(\Mresolution,\Tresolution).
\end{equation*}
It turns out that $U(C_j(\Mresolution))$, the set of all directions in $C_j(\Mresolution)$, provides a net for the directions of long chords on $\manifold$. In contrast, $C'_j(\Mresolution,\Tresolution)$ forms a net for the directions in $U(\manifold)$ that correspond to the short chords on $\manifold$. It is therefore not surprising that $\{\Ucover_j(\Uresolution)\}$ proves to be a sequence of increasingly finer covers for $U(\manifold)$. This discussion is formalized in the next lemma. We remark that Lemma~\ref{lemma:Tj is a refinement of U(M)} holds more generally for all constants $\Cr{eta},\Cr{eta'}$ that satisfy the conditions listed in the proof.

\begin{lemma}\label{lemma:Tj is a refinement of U(M)}
	Set $\Cr{eta}=0.4$ and $\Cr{eta'}=1.7-\sqrt{2}$. For every $j\in\ints$, $\Ucover_j(\Uresolution)$, as constructed above, is a $(2^{-j}\Uresolution)$-net for $U(\manifold)$, when $\Uresolution\le1/2$.
	Under the mild assumption that
	\begin{equation}\label{eq:mild assump on volume 1}
		\frac{\volume}{\condition^K}\ge \left(\frac{21}{2\sqrt{K}}\right)^K,
	\end{equation}
	it also holds that
	\begin{equation}
		\#\Ucover_j(\Uresolution) \le 2\cdot4^{2jK}\left(\frac{6.12\sqrt{K}}{\Uresolution^2}\right)^{2K}\left(\frac{\volume}{\condition^K}\right)^2=:{\Ucoverno}_j(\Uresolution).
	\end{equation}
\end{lemma}
\notechecked{used distinguishing font for $\Ucoverno_j$} \notechecked{proof moved to here}
\begin{proof}
Consider two arbitrary but distinct points $x_1,x_2\in\manifold$. For $\Cl{thresh}>0$ to be set later in the proof, we separate the treatment of long and short chords, i.e., $\|x_2-x_1\|/\condition> \Cr{thresh}\sqrt{\Mresolution/\condition}=:\thresh$ and $\|x_2-x_1\|/\condition\le \thresh$, and in this strategy we follow \cite{baraniuk2009random,clarkson2008tighter}. Short chords are distinct in that, as we will see later, they have to be approximated with nearby tangent vectors. For convenience, let us also define
$$
	U_\thresh^l(\manifold):=\{U(z_1,z_2):\|z_2-z_1\|>\thresh\condition,\, z_1,z_2\in\manifold\},
$$
$$
	U_\thresh^s(\manifold):=\{U(z_1,z_2):0<\|z_2-z_1\|\le\thresh\condition,\,z_1,z_2\in\manifold\}.
$$
Of course, $U_\thresh^l(\manifold)\cup U_\thresh^s(\manifold)=U(\manifold)$, although their intersection might not be empty.

Suppose that $\|x_2-x_1\|/\condition> \thresh=\Cr{thresh}\sqrt{\Mresolution/\condition}$ so that $U(x_1,x_2)\in U_\thresh^l(\manifold)$. Since $C_0(\Mresolution)$ is an $\Mresolution$-net for $\manifold$, there exist $p_1$ and $p_2$ in $C_0(\Mresolution)$ such that $\|x_1-p_1\|,\|x_2-p_2\|\le\Mresolution$. It then follows from Lemma \ref{lemma:implicit in Clarkson} (with $a_1=x_1$, $a_2=x_2$, $b_1=p_1$, and $b_2=p_2$) that 
\begin{equation}\label{eq:long chords proof p1}
	\|U(x_1,x_2)-U(p_1,p_2)\|\le \frac{4}{\Cr{thresh}}\sqrt{\frac{\Mresolution}{\condition}}=\frac{4\Cr{eta}\Uresolution}{\Cr{thresh}}.
\end{equation}
Now, assuming that
\begin{equation}\label{eq:covering cond 1}
	4\Cr{eta}=\Cr{thresh},
\end{equation}
and leveraging the fact that the choice of $x_1,x_2\in\manifold$ was arbitrary, we conclude that $U(C_0(\Mresolution))$ is a $\Uresolution$-net for $U_\thresh^l(\manifold)$.

On the other hand, suppose that $0<\|x_2-x_1\|/\condition\le \thresh=\Cr{thresh}\sqrt{\Mresolution/\condition}=4\Cr{eta}\sqrt{\Mresolution/\condition}$ so that $U(x_1,x_2)\in U_\thresh^s(\manifold)$. We assume that
\begin{equation}\label{eq:assump 1 for short chords}
	\frac{\Mresolution}{\condition}=\Cr{eta}^2\Uresolution^2 < \min\left(\frac{1}{64\Cr{eta}^2},\frac{1}{2}\right),
\end{equation}
so that, in particular, $\|x_2-x_1\|<\condition/2$. Since $C_0(\Mresolution)$ is an $\Mresolution$-net for $\manifold$, there exists a point $p\in C_0(\Mresolution)$ such that $\|x_1-p\|\le\Mresolution<\condition/2$. Lemma \ref{lemma:distance of a shorrt chord to a nearby tangent plane} (with $l_1=\Mresolution$ and $l_2=4\Cr{eta}\sqrt{\condition\Mresolution}$) then implies that the direction of the chord connecting $x_1$ to $x_2$ can be approximated with a tangent vector in $\tan_p$, that is
\begin{equation}\label{eq:short chords proof p1}
	\|U(x_1,x_2)-\downarrow_p U(x_1,x_2)\|\le \sqrt{\frac{2\Mresolution}{\condition}}+2\Cr{eta}\sqrt{\frac{\Mresolution}{\condition}}=\left(\sqrt{2}+2\Cr{eta}\right)\Cr{eta}\Uresolution.
\end{equation}
Recall that $C^{0,p}(\Tresolution)$ is an $\Tresolution$-net for the unit ball centered at $p$ and along $\tan_p$. So, there also exists a vector $v\in C^{0,p}(\Tresolution)$ such that $\|\downarrow_p U(x_1,x_2)-v\|\le \Tresolution=\Cr{eta}\Cr{eta'}\Uresolution$. Using the triangle inequality, we therefore arrive at
\begin{equation}\label{eq:short chords proof p2}
	\|U(x_1,x_2)-v\|\le \left(\sqrt{2}+2\Cr{eta}+\Cr{eta'}\right)\Cr{eta}\Uresolution.
\end{equation}
Assuming that
\begin{equation}\label{eq:covering cond 2}
	\sqrt{2}+2\Cr{eta}+\Cr{eta'}=\Cr{eta}^{-1},
\end{equation}
and leveraging the fact that the choice of $x_1,x_2\in\manifold$ was arbitrary, we conclude that $C'_0(\Mresolution,\Tresolution)$ is a $\Uresolution$-net for $U^s_\thresh(\manifold)$.
Overall, under \eqref{eq:covering cond 1}, \eqref{eq:assump 1 for short chords}, and \eqref{eq:covering cond 2}, $\Ucover_0(\Uresolution)=U(C_0(\Mresolution))\cup C_0'(\Mresolution,\Tresolution)$ is a $\Uresolution$-net for $U(\manifold)$. By repeating the argument above (with $\Mresolution,\Uresolution,\Tresolution,\thresh$ replaced with $\Mresolution/4^j,\Uresolution/2^j,\Tresolution/2^j,\thresh/2^j$) we observe that $\Ucover_j(\Uresolution)$ is a ($2^{-j}\Uresolution$)-net for $U(\manifold)$, for every $j\in\ints$. In particular, the choice of $\Cr{eta}=0.4,\Cr{eta'}=1.7-\sqrt{2},\Cr{thresh}=1.6$ satisfies the conditions above for every $\Uresolution\le1/2$ and completes the proof of the first statement in Lemma~\ref{lemma:Tj is a refinement of U(M)}.

In order to bound the cardinality of $\Ucover_j(\Uresolution)$, we begin with estimating $\#C'_j(\Mresolution,\Tresolution)$.
According to Lemma \ref{lemma:cover no of unit ball}, we can write  that
\begin{equation}
	\label{eq:cover no for short chords pre}
	\#C'_{j}\left(\Mresolution,\Tresolution\right)  \le\left(\frac{3\cdot2^j}{\Tresolution} \right)^{K}\cdot\#C_{j}\left(\Mresolution\right),
\end{equation}
which holds assuming that $\Tresolution\le 1$, i.e., $\Cr{eta}\Cr{eta'}\Uresolution\le 1$. (Our choice of $\Cr{eta},\Cr{eta'}$ above satisfies this condition.) It is possible now to write that
\begin{align}\label{eq:bound on Tj 1st}
	\#\Ucover_j(\Uresolution) & \le (\#C_j(\Mresolution))^2+\#C'_j(\Mresolution,\Tresolution)\nonumber\\
	& \le (\#C_j(\Mresolution))^2+\left(\frac{3\cdot2^j}{\Tresolution}\right)^K\cdot \#C_{j}\left(\Mresolution\right)\nonumber\\
	& \le 2\max\left(4^{jK}\coverbnd_0(\Mresolution),\left(\frac{3\cdot2^j}{\Tresolution}\right)^K\right)\cdot 4^{jK}\coverbnd_0(\Mresolution),
\end{align}
where we used \eqref{eq:cover no for short chords pre} in the second line and  \eqref{eq:from Cj to C0} in the last line.
To guarantee that the first term dominates the maximum in \eqref{eq:bound on Tj 1st}, it suffices (according to the definition of $\coverbnd_0(\Mresolution)$ in \eqref{eq:niyogi covering number 1}) to enforce that
\begin{equation*}
	\left(\frac{2}{\theta(\Mresolution/4\condition)\Mresolution}\right)^K\frac{\volume}{V_{\ball_K}}\ge\left(\frac{3}{\Tresolution}\right)^K,
\end{equation*}
which, after plugging in for $\Mresolution$ and $\Tresolution$ in terms of $\Uresolution$ and using the hypothesis that $\Uresolution\le 1$, is satisfied under the mild assumption that
\begin{equation}\label{eq:mild assump on volume 1 raw}
	\frac{\volume}{\condition^K}\ge2.5^KV_{\ball_K}\ge\left(\frac{3\Cr{eta}}{2\Cr{eta'}}\right)^K V_{\ball_K}.
\end{equation}
The assumption in \eqref{eq:mild assump on volume 1 raw} allows us to simplify \eqref{eq:bound on Tj 1st} and obtain that
\begin{equation}\label{eq:bound on T0 2nd}
	\#\Ucover_j(\Uresolution)\le 2\cdot4^{2jK} \coverbnd_0^2(\Mresolution).
\end{equation}
It follows from \eqref{eq:bound on T0 2nd} and the definition of $\coverbnd_0(\Mresolution)$ in \eqref{eq:niyogi covering number 1} that
\begin{align*}
	\#\Ucover_j(\Uresolution)\le 2\cdot4^{2jK}\left(\frac{2}{\theta(\Cr{eta}^2/4)\Cr{eta}^2\condition\Uresolution^2}\right)^{2K}\left(\frac{\volume}{V_{\ball_K}}\right)^2
	\le 2\cdot4^{2jK}\left(\frac{12.52}{\condition\Uresolution^2}\right)^{2K}\left(\frac{\volume}{V_{\ball_K}}\right)^2,
\end{align*}
where we used the fact that $\Uresolution\le1$. We remind the reader that
\begin{equation}
	\left(\frac{4\pi}{K+2}\right)^{K/2}\le V_{\ball_K}=\frac{\pi^{K/2}}{\Gamma\left(\frac{K}{2}+1\right)}\le\left(\frac{2e\pi}{K+2}\right)^{K/2},\label{eq:bnds on VkBk}
\end{equation}
where the inequalities follow from the fact that $\left(K/e\right)^{K-1}\le\Gamma\left(K\right)\le\left(K/2\right)^{K-1}$ for $K\in\mathbb{N}$ \cite{olver2010nist}. Here $\Gamma(\cdot)$ denotes the Gamma function. The above inequality leads us to
\begin{align*}
	\#\Ucover_j(\Uresolution) \le 2\cdot4^{2jK}\left(\frac{6.12\sqrt{K}}{\Uresolution^2}\right)^{2K}\left(\frac{\volume}{\condition^K}\right)^2,
\end{align*}
which holds under the mild assumption that $\volume/\condition^K\ge (21/\sqrt{K})^K$. Indeed, this assumption is obtained by plugging  our choice of $\Cr{eta},\Cr{eta'}$ into \eqref{eq:mild assump on volume 1 raw}. This completes the proof of Lemma~\ref{lemma:Tj is a refinement of U(M)}.
\end{proof}

\subsection{Applying the chaining argument}
\label{sub:Applying-the-Chaining arg}

Every $y\in U(\manifold)$ can be represented with a chain of points in $\{\Ucover_j(\Uresolution)\}$. Let $\pi_j(y)$ be the nearest point to $y$ in $\Ucover_j(\Uresolution)$. This way we obtain a sequence $\{\pi_j(y)\} $ that represents $y$ via an almost surely convergent telescoping sum, that is
\begin{eqnarray}
	y=\pi_{0}\left(y\right)+\sum_{j\in\ints} \left(\pi_{j+1}(y)-\pi_{j}(y)\right).
	\label{eq:telescope}
\end{eqnarray}
Note that, for every $j\in\ints$ and every $y\in\manifold$, the length of the chord connecting $\pi_j(y)$ to $\pi_{j+1}(y)$ is no longer than $2^{-j+1}\Uresolution$. We are now ready to state a generic chaining argument that allows us to bound the failure probability of obtaining a stable embedding of $\manifold$ in terms of its geometrical properties. The interested reader is referred to \cite{talagrand2005generic} for more information about the generic chaining.

\begin{lemma}
	\label{lemma:chaining}
	Fix $0<\Uresolution<\epsilon_1<\epsilon\le1/3$, and  $\epsilon_{2}>0$ such that $\epsilon_{1}+\epsilon_{2}=\epsilon$. Choose $\Cl{eps1 n delta},\Cl{eps2 n j}>0$ so that $\epsilon_{1}/\Uresolution\ge\frac{1+\Cr{eps1 n delta}}{1-\Cr{eps1 n delta}}$ and  $\epsilon_{2}/\Uresolution\ge \Cr{eps2 n j}$.
	Then, under \eqref{eq:mild assump on volume 1}, we have that
	\begin{align}
		\pr{ \sup_{y\in U\left(\manifold\right)}\left|\left\Vert \proj y\right\Vert -1\right|>\epsilon} &
		\le 2\Ucoverno_0(\Uresolution)\cdot\max_{t_{0}\in \Ucover_{0}\left(\Uresolution\right)}\pr{ \left|\left\Vert \proj t_{0}\right\Vert -\|t_0\|\right|>\Cr{eps1 n delta}\epsilon_1\|t_0\|}\nonumber\\
		& +2\sum_{j\in\ints} \Ucoverno_{j+1}^2(\Uresolution)\cdot\max_{(t_j,s_j)\in \links_j(\Uresolution)}\pr{ \left\Vert \proj s_{j}-\proj t_{j}\right\Vert >8^{-1}\Cr{eps2 n j}(j+1)\|s_j-t_j\|},
		\label{eq:chaining_lemma}
	\end{align}
	where $\{\Ucoverno_j(\Uresolution)\} $ were previously defined in Lemma \ref{lemma:Tj is a refinement of U(M)}. For $j\in\ints$, $\links_j(\Uresolution)$ is defined as
	\begin{equation*}
		\links_j(\Uresolution):=\left\{(t_j,s_j)\,:\,\pi_j(y)=t_j\mbox{ and }\pi_{j+1}(y)=s_j \mbox{ for some } y\in U(\manifold)\right\}.
	\end{equation*}
\end{lemma}
\notechecked{proof moved to here}
\begin{proof}
For notational convenience, let us denote the infinite sum in \eqref{eq:telescope} by $\Sigma(y)$. Then, using the triangle inequality, we observe that
\begin{align*}
	\pr{ \sup_{y\in U(\manifold)}\left\Vert \proj y\right\Vert>1+\epsilon} &
	=\pr{ \sup_{y}\left\Vert  \proj\pi_{0}\left(y\right)+\proj\Sigma(y)\right\Vert >1+\epsilon_{1}+\epsilon_{2}} \\
	& \le \pr{ \sup_{y}\left\Vert  \proj\pi_{0}\left(y\right)\right\Vert+\sup_y\left\Vert\proj\Sigma(y)\right\Vert >1+\epsilon_{1}+\epsilon_{2}} \\
	& \le\pr{ \sup_{y}\left\Vert \proj\pi_{0}\left(y\right)\right\Vert -1>\epsilon_{1}} +\pr{ \sup_{y}\left\Vert \proj\Sigma(y)\right\Vert >\epsilon_{2}},
\end{align*}
and similarly,
\begin{align*}
	\pr{ \inf_{y\in U(\manifold)}\left\Vert \proj y\right\Vert <1-\epsilon}
	& =\pr{ \inf_{y}\left\Vert \proj\pi_{0}\left(y\right)+\proj\Sigma(y)\right\Vert <1-\epsilon_{1}-\epsilon_{2}} \\
	& \le \pr{ \inf_{y}\left\Vert \proj\pi_{0}\left(y\right)\right\Vert-\sup_y\left\Vert\proj\Sigma(y)\right\Vert <1-\epsilon_{1}-\epsilon_{2}} \\
	& \le\pr{ \sup_{y}\,1-\left\Vert \proj\pi_{0}\left(y\right)\right\Vert >\epsilon_{1}} +\pr{ \sup_{y}\left\Vert \proj\Sigma(y)\right\Vert >\epsilon_{2}}.
\end{align*}
We can therefore argue that
\begin{align}
	\pr{ \sup_{y\in U(\manifold)}\left|\left\Vert \proj y\right\Vert-1\right|>\epsilon}
	&  \le \pr{ \sup_{y}\left\Vert \proj y\right\Vert>1+\epsilon}+\pr{ \inf_{y}\left\Vert \proj y\right\Vert <1-\epsilon}  \nonumber\\
	& \le2\pr{ \sup_{y}\left|\left\Vert \proj\pi_{0}\left(y\right)\right\Vert -1\right|>\epsilon_{1}} +2\pr{ \sup_{y}\left\Vert \proj\Sigma(y)\right\Vert >\epsilon_{2}}.
	\label{eq:build back from up & down}
\end{align}
 Consider the first probability on the last line of \eqref{eq:build back from up & down}:
\begin{align*}
	\pr{ \sup_{y\in U(\manifold)}\left|\left\Vert \proj\pi_{0}\left(y\right)\right\Vert -1\right|>\epsilon_{1}}
	& \le \pr{ \sup_{y}\left|\left\Vert \proj \pi_0(y)\right\Vert -\|\pi_0(y)\|\right|+\sup_{y}\left|\|\pi_0(y)\|-1\right|>\epsilon_{1}} \\
	& \le \pr{ \sup_{y}\left|\left\Vert \proj \pi_0(y)\right\Vert -\|\pi_0(y)\|\right|>\epsilon_{1}-\Uresolution} \\
	& \le \pr{ \sup_{y}\frac{\left|\left\Vert \proj \pi_0(y)\right\Vert -\|\pi_0(y)\|\right|}{\|\pi_0(y)\|}>\frac{\epsilon_{1}-\Uresolution}{1+\Uresolution}} \\
	& \le \pr{ \sup_{y}\frac{\left|\left\Vert \proj \pi_0(y)\right\Vert -\|\pi_0(y)\|\right|}{\|\pi_0(y)\|}>\Cr{eps1 n delta}\epsilon_1} \\
	& \le \pr{ \max_{t_{0}\in \Ucover_{0}(\Uresolution)}\frac{\left|\left\Vert \proj t_0\right\Vert -\|t_0\|\right|}{\|t_0\|}>\Cr{eps1 n delta}\epsilon_1} \\
	& \le\#\Ucover_{0}\left(\Uresolution\right)\cdot\max_{t_{0}\in \Ucover_{0}(\Uresolution)}\pr{ \left|\left\Vert \proj t_{0}\right\Vert -\|t_0\|\right|>\Cr{eps1 n delta}\epsilon_1\|t_0\|}
\end{align*}
where the first line uses the triangle inequality. The second and third  lines hold on account of $\Ucover_0(\Uresolution)$ being a net for a subset of $\mathbb{S}^{N-1}$, namely $U(\manifold)$. An application of the union bound gives the last line above.

Now consider the second probability on the last line of \eqref{eq:build back from up & down}. By the definition of $\Sigma(y)$, we observe that
\begin{align*}
	& \pr{ \sup_{y\in U(\manifold)}\left\Vert \proj\Sigma(y)\right\Vert >\epsilon_{2}} \\
	& =\pr{ \sup_{y}\left\Vert \sum \proj\pi_{j+1}(y)-\proj\pi_{j}(y)\right\Vert >\epsilon_{2}} \\
	& \le\pr{ \sum_j \max_{(t_j,s_j)\in \links_j(\Uresolution)}\left\Vert \proj s_{j}-\proj t_{j}\right\Vert >\Cr{eps2 n j}\Uresolution} \\
	& = \pr{ \sum_j \max_{(t_j,s_j)\in \links_j(\Uresolution)}\left\Vert \proj s_{j}-\proj t_{j}\right\Vert >\Cr{eps2 n j}\sum_j (j+1)2^{-j-2}\Uresolution} \\
	& \le\sum_j \pr{ \max_{(t_j,s_j)\in \links_j(\Uresolution)}\left\Vert \proj s_{j}-\proj t_{j}\right\Vert >\Cr{eps2 n j}(j+1)2^{-j-2}\Uresolution} \\
	& \le\sum_j \pr{ \max_{(t_j,s_j)\in \links_j(\Uresolution)}\left\Vert \proj s_{j}-\proj t_{j}\right\Vert >8^{-1}\Cr{eps2 n j}(j+1)\|s_j-t_j\|}\\
	& \le \sum_j \#T^2_{j+1}(\Uresolution)\max_{(t_j,s_j)\in \links_j(\Uresolution)}\pr{ \left\Vert \proj s_{j}-\proj t_{j}\right\Vert >8^{-1}\Cr{eps2 n j}(j+1)\|s_j-t_j\|}.
\end{align*}
The third line above uses the triangle inequality and the assumption on $\epsilon_2$, while the fifth and last lines use the union bound. It can be easily verified that the infinite sum on the right hand side of the inequality in the fourth line equals one. \notechecked{added the next sentence.} In the sixth line, we used the observation that $(t_{j},s_{j})\in Q_{j}(\delta)$ implies that $\|s_{j}-t_{j}\|\le2^{-j}\delta+2^{-j-1}\delta\le2^{-j+1}\delta$.
Having upper bounds for both terms on the last line of \eqref{eq:build back from up & down}, we overall arrive at
\begin{align*}
	\pr{ \sup_{y\in U(\manifold)}\left|\left\Vert \proj y\right\Vert-1\right|>\epsilon} & \le2\#\Ucover_{0}\left(\Uresolution\right)\cdot\max_{t_{0}\in \Ucover_{0}(\Uresolution)}\pr{ \left|\left\Vert \proj t_{0}\right\Vert -\|t_0\|\right|>\Cr{eps1 n delta}\epsilon_1\|t_0\|}\\
	& +2\sum_j \#T^2_{j+1}\left(\Uresolution\right)\max_{{(t_j,s_j)\in \links_j(\Uresolution)}}\pr{ \left\Vert \proj s_{j}-\proj t_{j}\right\Vert >8^{-1}\Cr{eps2 n j}(j+1)\|s_j-t_j\|}.
\end{align*}
From Lemma \ref{lemma:Tj is a refinement of U(M)}, $\#\Ucover_j(\Uresolution)\le \Ucoverno_j(\Uresolution)$. This establishes Lemma \ref{lemma:chaining}.
\end{proof}

There are two type of probabilities involved in the upper bound above. One controls the large deviations of $\|\proj t_0\|$ from its expectation, and the other corresponds to very large (one sided) deviations of $\|\proj s_j -\proj t_j\|$ from its expectation.
As claimed in the next lemma and proved in Appendix \ref{sec:Proof-of-Lemma basics for fixed y}, both of these probabilities are exponentially small when $M$ is large enough.
\begin{lemma}
	\label{lemma:basics for fixed}
	Fix $0\le\lambda\le1/3$ and $\lambda'\ge1/5$. Then, for fixed $y\in\real^N$, we have \notechecked{replaced $y_1$ with $y$ in first inequality below; replaced $y_2-y_1$ with $y$ in second inequality below}
	\begin{eqnarray}
		\pr{ \left|\left\Vert \proj y\right\Vert -\|y\|\right|>\lambda\|y\|} \le2e^{-\frac{M\lambda^{2}}{6}}
		\label{eq:COM_norm}\\
		\pr{ \left\Vert \proj y\right\Vert >\left(1+\lambda'\right)\left\Vert y\right\Vert } \le e^{-\frac{M\lambda'}{7}}.\label{eq:COM_diff}
	\end{eqnarray}
 \end{lemma}

Now fix $\epsilon\le1/3$ and set $\epsilon_1=9\epsilon/10$. Taking  $\Cr{eps1 n delta}=\sqrt{6/7}$, $\Cr{eps2 n j}=16$, $\Uresolution=\epsilon/160$ and finally assuming \eqref{eq:mild assump on volume 1} guarantees that Lemma \ref{lemma:Tj is a refinement of U(M)} is in force. Under this setup, note that an upper bound for the first term on the right hand side of \eqref{eq:chaining_lemma} can be found by applying \eqref{eq:COM_norm} (after plugging in for  $\Cr{eps1 n delta}$):
\begin{align*}
	2\Ucoverno_0(\Uresolution)\cdot\max_{t_{0}\in \Ucover_{0}(\Uresolution)}\pr{ \left|\left\Vert \proj t_{0}\right\Vert -\|t_0\|\right|>\sqrt{\frac{6}{7}}\epsilon_1\|t_0\|}  & \le 2\Ucoverno_{0}(\Uresolution)\cdot2e^{-\frac{ M\epsilon_{1}^{2}}{7}},
\end{align*}
and, assuming that \notechecked{corrected notation below by changing $N_0(\Uresolution)$ to $\Ucoverno_0(\Uresolution)$}
\begin{equation}
	M\ge14\epsilon_{1}^{-2}\log \Ucoverno_0(\Uresolution),
	\label{eq:chain cond on a 1st}
\end{equation}
we arrive at
\begin{align}
	2\Ucoverno_0(\Uresolution)\cdot\max_{t_{0}\in \Ucover_{0}(\Uresolution)}\pr{ \left|\left\Vert \proj t_{0}\right\Vert -\|t_0\|\right|>\sqrt{\frac{6}{7}}\epsilon_{1}\|t_0\|}  & \le4e^{-\frac{M\epsilon_{1}^{2}}{14}}.
	\label{eq:bnd on 1nd term of chaining lemma}
\end{align}
In order to bound the second term on the right hand side of \eqref{eq:chaining_lemma}, we
proceed as follows. \notechecked{next sentence edited} Consider the maximum inside the summation. After plugging in for $\Cr{eps2 n j}$ and applying \eqref{eq:COM_diff},  we can bound this maximum as
\begin{equation*}
	\max_{(t_j,s_j)\in \links_j(\Uresolution)}\pr{ \left\Vert \proj s_{j}-\proj t_{j}\right\Vert >2(j+1)\|s_j-t_j\|}\le e^{-\frac{(2j+1)M}{7}}.
\end{equation*}
Using the  estimate above and Lemma \ref{lemma:Tj is a refinement of U(M)}, we get an upper bound for the second term on the right hand side of \eqref{eq:chaining_lemma}: \notechecked{removed duplicate mention of constant}
\begin{align}
	& 2\sum_{j\in\ints} \Ucoverno^2_{j+1}(\Uresolution)\cdot\max_{(t_j,s_j)\in \links_j(\Uresolution)}\pr{ \left\Vert \proj s_{j}-\proj t_{j}\right\Vert >2(j+1)\|s_j-t_j\|}\nonumber\\
	& \le 2\Ucoverno_0^2(\Uresolution)e^{-\frac{M}{7}} 4^{4K}\sum_{j\in\ints}  4^{4jK} e^{-\frac{2}{7}jM}.\label{eq:bnd on 2nd term of chaining pre}
\end{align}
Assuming that
\begin{equation}\label{eq:to be absorbed cond on M}
	M\ge \max(32\log \Ucoverno_0(\Uresolution),310K)
\end{equation}
allows us to continue simplifying \eqref{eq:bnd on 2nd term of chaining pre}, therefore arriving at
\begin{align}
	& 2\sum \Ucoverno_{j+1}^2(\Uresolution)\cdot\max_{(t_j,s_j)\in \links_j(\Uresolution)}\pr{ \left\Vert \proj s_{j}-\proj t_{j}\right\Vert >2(j+1)\|s_j-t_j\|}\nonumber\\
	& \le 4e^{-\frac{M}{17}}.\label{eq:bnd on 2nd term of chaining lemma}
\end{align}
We can now combine \eqref{eq:bnd on 1nd term of chaining lemma} and \eqref{eq:bnd on 2nd term of chaining lemma} to obtain
 \[
	\pr{ \sup_{y}\left|\left\Vert \proj y\right\Vert-1\right|>\epsilon} \le4e^{-\frac{M\epsilon_{1}^{2}}{14}}+4e^{-\frac{M}{17}}\le8e^{-\frac{M\epsilon_{1}^{2}}{14}},
\]
where the second inequality follows since $\epsilon_1\le1/3$ and thus $\epsilon_{1}^{2}/14\le1/17$. In particular, to achieve a failure probability of at most $\failprob\le1$,
we need
\begin{equation}
      M\ge14\epsilon_{1}^{-2}\log\left(8/\failprob\right).\label{eq:chain cond on m 2nd}
\end{equation}
\notechecked{this paragraph is edited} Assuming that \eqref{eq:mild assump on volume 1} holds and that $\epsilon\le1/3$,  we verify that \eqref{eq:to be absorbed cond on M} may be absorbed into \eqref{eq:chain cond on a 1st} (i.e., \eqref{eq:chain cond on a 1st} implies \eqref{eq:to be absorbed cond on M}). We are now left with \eqref{eq:chain cond on a 1st} and \eqref{eq:chain cond on m 2nd}, which are in turn lumped into a single lower bound on $M$ (after plugging in for $\Uresolution$), that is \notechecked{corrected notation below by changing $N_0(\Uresolution)$ to $\Ucoverno_0(\Uresolution)$, and changed $6$ to $6.12$}
\begin{align}
	M & \ge18\epsilon^{-2}\max\left(\log(2\volume^{2})+24K+2K\log\left(\frac{\sqrt{K}}{\condition\epsilon^{2}}\right),\log\left(\frac{8}{\failprob}\right)\right)\nonumber\\
	& \ge 18\epsilon^{-2}\max\left(\log\left(2\volume^{2}\left(\frac{6.12\sqrt{K}}{\condition\Uresolution^{2}}\right)^{2K}\right),\log\left(\frac{8}{\failprob}\right)\right)\nonumber\\
	& = 18\epsilon^{-2}\max\left(\log \Ucoverno_0(\Uresolution),\log\left(8/\rho\right)\right)\nonumber\\
	& \ge 14\epsilon_{1}^{-2}\max\left(\log \Ucoverno_0(\Uresolution),\log\left(8/\rho\right)\right).\label{eq:final bnd on M pre}
\end{align}
Therefore, we proved that
\begin{equation*}
	\pr{\sup_{y\in U(\manifold)}|\|\proj y\|-1|>\epsilon}\le\failprob,
\end{equation*}
provided that $M$ satisfies \eqref{eq:final bnd on M pre}. This completes the proof of Theorem~\ref{thm:new manifold embedding}.

\section{Proof of Lemma \ref{lemma:basics for fixed}}
\label{sec:Proof-of-Lemma basics for fixed y}

\notechecked{updated proof, replacing $y_1$ (in the first inequality) and $y_2-y_1$ (in the second inequality) with $y$} The proof is elementary. It is easily verified that $\mathbb{E}\left\Vert \proj y\right\Vert ^{2}=\|y\|^2$, and we then note that
\begin{align*}
	\pr{\left|\left\Vert \proj y\right\Vert -\|y\|\right|>\lambda\|y\|}
	& =\pr{ \left\Vert \proj y\right\Vert >(1+\lambda)\|y\|} +\pr{ \left\Vert \proj y\right\Vert <(1-\lambda)\|y\|} \\
	& \le\pr{ \left\Vert \proj y\right\Vert ^{2}>(1+\lambda)\|y\|^2} +\pr{ \left\Vert \proj y\right\Vert ^{2}<(1-\lambda)\|y\|^2} \\
	& \le2e^{-\frac{M}{2}\left(\frac{\lambda^{2}}{2}-\frac{\lambda^{3}}{3}\right)}\\
	& \le2e^{-\frac{M\lambda^{2}}{2}\left(\frac{1}{2}-\frac{1}{9}\right)}\\
	& \le2e^{-\frac{M\lambda^{2}}{6}},
\end{align*}
where the third line uses a well-known concentration bound \cite{achlioptas2001database}. The fourth line holds because $\lambda\le1/3$. This establishes the first inequality in Lemma \ref{lemma:basics for fixed}. For the second inequality, assume, without loss of generality, that $\|y\|=1$. We begin by observing that
\begin{align}
	\pr{ \left\Vert \proj y\right\Vert >\left(1+\lambda'\right)\left\Vert y\right\Vert }  & =\pr{ \left\Vert \proj y\right\Vert >1+\lambda'} \nonumber\\
	& \le\pr{ \left\Vert \proj y\right\Vert ^{2}>1+2\lambda'}\nonumber \\
	& =\pr{ M^{-1}\sum_{i=1}^{M}n_{i}^{2}-1>2\lambda'}\nonumber \\
	& =\pr{ \sum_{i=1}^{M}n_{i}^{2}-M>2\lambda' M},\label{eq:2nd ineq in basics 1st attempt}
\end{align}
where $n_{1},n_{2},\cdots,n_{M}$ are zero-mean and unit-variance Gaussian random variables. The third line above follows since the entries of the vector $\proj y$ are distributed as i.i.d.\ zero-mean Gaussians with variance of $1/M$. We now recall Lemma 1 in \cite{laurent2000adaptive}, which states that
\begin{equation}
	\pr{ \sum_{i=1}^{M}n_{i}^{2}-M>2\sqrt{M\alpha}+2\alpha} \le e^{-\alpha},
	\label{eq:large deviation COM for gaussians}
\end{equation}
for $\alpha>0$.  Comparing the last line in \eqref{eq:2nd ineq in basics 1st attempt} to the inequality above, we observe that taking
\[
	\alpha=\frac{M}{4}\left(\sqrt{1+4\lambda'}-1\right)^2
\]
allows us to continue simplifying \eqref{eq:2nd ineq in basics 1st attempt} to obtain that
\begin{align}
	\pr{ \left\Vert \proj y\right\Vert >\left(1+\lambda'\right) }   \le 	\pr{ \sum_{i=1}^{M}n_{i}^{2}-M>2\sqrt{M\alpha}+2\alpha} \le e^{-\alpha}.
	\label{eq:2nd ineq in basics 2nd attempt}
\end{align}
It is easily verified that $\sqrt{1+4\lambda'}-1\ge (3-\sqrt{5})\sqrt{\lambda'}$ when $\lambda'\ge1/5$. It follows that
\begin{align}
	\alpha &\ge\frac{M}{4}\cdot (3-\sqrt{5})^2\lambda' \ge M\lambda'/7,
\end{align}
and consequently,
\[
	\pr{ \left\Vert \proj y\right\Vert>\left(1+\lambda'\right)\left\Vert y\right\Vert }\le e^{-\frac{M\lambda'}{7}},
\]
as claimed. This establishes the second inequality in Lemma \ref{lemma:basics for fixed} and completes the proof.

\section{Proof of Theorem~\ref{thm:kappa}} \label{app:kappa}

\notechecked{adjusted proof to account for noise} Fix $\alpha \in [1-\epsilon, 1+\epsilon]$. We consider any two
points $w_a,w_b \in \manifold$ such that
$$
\frac{\dist{\proj w_a}{\proj w_b}}{\dist{w_a}{w_b}} = \alpha,
$$
and supposing that $x$ is closer to $w_a$, i.e.,
$$
\dist{x}{w_a} \le \dist{x}{w_b},
$$
but $y=\proj x+n$ is closer to $\proj w_b$, i.e.,
$$
\dist{y}{\proj w_b} \le \dist{y}{\proj w_a},
$$
we seek the maximum value that $
\dist{x}{w_b}
$
may take. In other words, we wish to bound the worst possible
``mistake'' (according to our error criterion) between two candidate
points on the manifold whose distance is scaled by the factor
$\alpha$.

This can be posed in the form of an optimization problem
\begin{eqnarray*}
\max_{x \in \real^\dim,  w_a,w_b \in \manifold}
\dist{x}{w_b} &\mathrm{s.t.}&\dist{x}{w_a} \le
\dist{x}{w_b},\\
&&
\dist{y}{\proj w_b} \le \dist{y}{\proj w_a},\\
&&\frac{\dist{\proj w_a}{\proj w_b}}{\dist{w_a}{w_b}} = \alpha.
\end{eqnarray*}
For simplicity, we may expand the constraint set to include all
$w_a,w_b \in \real^\dim$; the solution to this larger problem is an
upper bound for the solution to the case where $w_a,w_b \in
\manifold$.
%
%
This leaves
\begin{eqnarray*}
\max_{x,w_a,w_b \in \real^\dim} \dist{x}{w_b}
&\mathrm{s.t.}&\norm{y-\proj w_b} \le \norm{y-\Phi w_a},\\
&&\frac{\norm{\proj w_b-\Phi w_a}}{\norm{w_b-w_a}} = \alpha.
\end{eqnarray*}
where we also ignored the first constraint (because of its relation
to the objective function). Under the constraints above, the objective satisfies
\begin{align*}
	\|x-w_b\|& \le \|x-w_a\|+\|w_b-w_a\|\\
	& = \|x-w_a\|+\|\proj w_b-\Phi w_a\|/\alpha\\
	& \le  \|x-w_a\|+2 \|y-\Phi w_a\|/\alpha\\
	& \le  \|x-w_a\|+2 (\|\Phi x-\Phi w_a\|+\|n\|)/\alpha\\
	& \le \|x-w_a\|+\frac{2\sing_M(\proj)}{1-\epsilon}\|x-w_a\|+\frac{2\|n\|}{1-\epsilon}\\
	& \le \frac{1}{1-\epsilon}\left(2\sing_M(\proj)+1\right)\|x-w_a\|+\frac{2\|n\|}{1-\epsilon}.
\end{align*}
The first line follows from the triangle inequality. The first identity above uses the first constraint. The first constraint (via the triangle inequality) implies that $\|\proj w_b-\Phi w_a\|\le 2 \|y-\Phi w_a\|$ and the third line thus follows. The fourth line uses the triangle inequality one more time. The fifth line follows after considering the possible range of $\alpha$.
To reiterate, the above conclusion holds for any observation $x$ that could be mistakenly paired with $w_b$
instead of $w_a$ (under a  $\proj$ that scales the distance
$\dist{w_a}{w_b}$ by $\alpha$). This completes the proof of Theorem~\ref{thm:kappa} after noting that $(1-\epsilon)^{-1}\le 1+2\epsilon$ when $\epsilon\le 1/2$.

\section{Proof of Proposition~\ref{lemma:l2 l2 bound is sharp}}
\label{sec:Proof of l2 l2 bound is sharp}

Set
$$
\dev := (1+\epsilon)\left(\sing_m(\proj)\right)^{-1},
$$
and let
$$
x = e_1+\dev u,
$$
where $\|u\|\le1$ belongs to the row span of $\proj$ and satisfies $\proj x=\proj(e_1+\dev u)=0$.
Finding such $u$ is possible because
$$
\|\proj e_1\|\le (1+\epsilon)\|e_1\|=1+\epsilon= \dev \cdot\sing_m(\proj)=\dev\cdot \sing_m(\proj) \|v\| \le \dev\|\proj v\|,
$$
for every unit vector $v$ in the row span of $\proj$. The first inequality holds because $e_1,{\bf 0}\in\manifold$ and $\proj$ stably embeds $\manifold$. The second equality holds by our choice of $\dev$, and the last inequality holds because $v$ belongs to the row span of $\proj$. With our choice of $x$ above, we have $\proj x=0$ and therefore $\widehat{x}=0$. On the other hand,
$$
\|x-x^\ast\|\le \|x-e_1\|\le\dev.
$$
It follows that
$$
\frac{\|x-\widehat{x}\|}{\|x-x^\ast\|}\ge \frac{\|x-\widehat{x}\|}{\|x-e_1\|} \ge \frac{\|x\|}{\dev} \ge \frac{1-\dev}{\dev}\ge \frac{1}{2\dev}
 =\frac{1}{2(1+\epsilon)}\sing_m(\proj).
$$
Indeed, one can verify that $\dev\le 1/2$ because (by hypothesis) $\epsilon\le1/3$ and $\sing_m(\proj)\ge8/3$. This immediately implies the second to last inequality above. This completes the proof of Proposition~\ref{lemma:l2 l2 bound is sharp}.

\section{Proof of Theorem~\ref{theo:bound3}}
\label{sec:proof of prob par est theorem}

\notechecked{replaced some instances of $\manifold$ with ``the manifold''} Our success in stably embedding $\manifold$ via random linear measurements (and what distinguished Theorem~\ref{thm:new manifold embedding} from embedding of a point cloud) relied on the smoothness of the manifold. This assumption enabled us to control the behavior of short chords on $\manifold$. However, $x$ does not generally belong to the manifold and hence, in general, we cannot control the direction of short chords connecting $x$ to $\manifold$. To deal with this issue, we proceed as follows. 
For fixed $\thresh>0$ to be specified later, define
\[
	\manifold_{\thresh}:=\left\{ z\in \manifold\,:\,\left\Vert z-x\right\Vert >\thresh\condition\right\},
\]
and let $\manifold_{\thresh}^{C}:=\manifold\backslash \manifold_{\thresh}$, i.e., the complement of $\manifold_\thresh$ in $\manifold$. Note that one of the two sets may be empty. Our first step towards a proof is to show that, for every $z\in \manifold_{\thresh}$ with an appropriately chosen $\thresh$, we have
\begin{equation}
	\left(1-\epsilon\right)\left\Vert z-x\right\Vert \le\left\Vert \proj z-\proj x\right\Vert \le\left(1+\epsilon\right)\left\Vert z-x\right\Vert.
	\label{eq:stable emb of U(M,z)}
\end{equation}
In other words, we first study the stable
embedding of the directions of all the chords connecting $x$ to $\manifold_{\thresh}$, namely $U(\manifold_\thresh,x)$, for an appropriate $\thresh$. This is addressed next.
\begin{lemma}
\label{lemma:stable embed of U(M,z)}
    \notechecked{eliminated some large equations and replaced with references, in lemma and its proof}
	Choose $0<\epsilon\le1/3$ and $0<\rho<1$. Conveniently assume that \eqref{eq:conv assump for par est} holds.
	If
	\begin{equation}
	M\ge18\epsilon^{-2}\max\left(11K+K\log\left(\frac{\sqrt{K}}{\condition\epsilon^{2}}\right)+\log\volume\,,\,\log\left(\frac{8}{\failprob}\right)\right),
    \label{eq:Mboundinproof}
	\end{equation}
	then, except with a probability of at most $\failprob$, \eqref{eq:stable emb of U(M,z)}
	holds for every $z\in \manifold_{\epsilon/40}$.
\end{lemma}

\begin{proof}
The proof strategy is identical to that in Appendix~\ref{sec:Proof of new manifold embedding}. We will prove that \eqref{eq:stable emb of U(M,z)} holds for every $z\in\manifold_{\epsilon/40}$, with high probability and provided that $M$ is large enough. As before, this is achieved by finding an upper bound on
\begin{equation}\label{eq:stable embed of U(M,x) goal}
	\pr{ \sup_{y\in U(\manifold_{\epsilon/40},x)}\left|\left\Vert \proj y\right\Vert -1\right|>\epsilon},
\end{equation}
for $\epsilon\le1/3$.

We begin again by constructing a sequence of increasingly finer covers for $U\left(\manifold_{\thresh},x\right)$, with $\thresh$ to be set later. We denote this sequence by $\{\Uparcover_j\left(\Uresolution\right)\}$---each $\Uparcover_j(\Uresolution)$ is a ($2^{-j}\Uresolution$)-net for $U(\manifold_{\thresh},x)$.
For $0<\Uresolution\le1/\sqrt{2}$, set $\Mresolution=\Uresolution^2\condition$ and $\thresh=4\Uresolution$. We form $\{\Uparcover_j\left(\Uresolution\right)\} $ from $\{C_{j}\left(\Mresolution\right)\} $, the sequence of covers for $\manifold$ constructed in Appendix~\ref{sec:Constructing a Sequence of Covers for U(M)}.
Indeed, the same argument in that section proves that $U(C_j(\Mresolution),x) $ is a ($2^{-j}\Uresolution$)-cover for $U(\manifold_\thresh,x)$. It also holds that
\begin{equation}
	\#\Uparcover_j(\Uresolution)\le\#C_j(\Mresolution)\le 4^{jK}\left(\frac{2}{\theta(\Uresolution^2/4)\Uresolution^2\condition}\right)^{K}\frac{\volume}{V_{\ball_K}}\le 4^{jK}\left(\frac{\sqrt{K}}{\Uresolution^2\condition}\right)^K\volume=:\Uparcoverno_j(\Uresolution).
\end{equation}
As before, we can represent every $y\in U(M_{\thresh},x)$ with an
infinite chain of points from the sequence of covers $\{\Uparcover_j\left(\Uresolution\right)\} $. After setting $\Uresolution=\epsilon/160$, using the same argument as the one in Appendix~\ref{sub:Applying-the-Chaining arg}, and exploiting the estimates above, one can verify that the failure probability in \eqref{eq:stable embed of U(M,x) goal} is at most $\failprob$, provided that \eqref{eq:Mboundinproof} holds.
\end{proof}

We now combine Lemma \ref{lemma:stable embed of U(M,z)} and
an elementary argument to complete the proof of Theorem~\ref{theo:bound3}.
It is possible to recognize two different cases: when $\widehat{x}\in \manifold_{\epsilon/40}^{C}$ and when $\widehat{x}\in \manifold_{\epsilon/40}$.
Clearly,
\begin{equation}
	\left\Vert x-x^{*}\right\Vert \le\left\Vert x-\widehat{x}\right\Vert \le\frac{\epsilon\condition}{40},\qquad\mbox{when }\widehat{x}\in \manifold_{\epsilon/40}^{C}.
	\label{eq:zhat bnd case 1}
\end{equation}
If, however, $\widehat{x}\in \manifold_{\epsilon/40}$, then a more detailed analysis is required. An application of Lemma~\ref{lemma:basics for fixed} implies that \eqref{eq:stable emb of U(M,z)} holds for $z=x^*$, except with a probability of at most $\failprob$ and provided that $M\ge6\epsilon^{-2}\log(2/\failprob)$. 
Suppose the assumptions in Lemma \ref{lemma:stable embed of U(M,z)} are met. Therefore, \eqref{eq:stable emb of U(M,z)} holds for every $z\in
\manifold_{\epsilon/40}\cup\{x^{*}\}$, except with a probability of at most $2\rho$. Also, by the definition of $x^{*}$ and $\widehat{x}$,  it holds true that
\[
	\left\Vert x-x^{*}\right\Vert \le\left\Vert x-\widehat{x}\right\Vert ~~~\text{and}~~~
	\left\Vert \left(\proj x+\noise\right)-\proj\widehat{x}\right\Vert \le\left\Vert \left(\proj x+\noise\right)-\proj x^{*}\right\Vert.
\]
Now, combining all these bounds and using several applications of the triangle inequality we obtain that \notechecked{removed a duplicate line below}
\begin{align*}
	\left\Vert x-\widehat{x}\right\Vert  & \le\left(1-\epsilon\right)^{-1}\left\Vert \proj x-\proj\widehat{x}\right\Vert \\
	& \le\left(1-\epsilon\right)^{-1}\left\Vert \left(\proj x+\noise\right)-\proj\widehat{x}\right\Vert +\left(1-\epsilon\right)^{-1}\left\Vert \noise\right\Vert \\
	& \le\left(1-\epsilon\right)^{-1}\left\Vert \left(\proj x+\noise\right)-\proj x^{*}\right\Vert +\left(1-\epsilon\right)^{-1}\left\Vert \noise\right\Vert \\
	& \le\left(1-\epsilon\right)^{-1}\left\Vert \proj x-\proj x^{*}\right\Vert +2\left(1-\epsilon\right)^{-1}\left\Vert \noise\right\Vert \\
	& \le\frac{1+\epsilon}{1-\epsilon}\left\Vert x-x^{*}\right\Vert +2\left(1-\epsilon\right)^{-1}\left\Vert \noise\right\Vert.
\end{align*}
\notechecked{corrected the next two inequalities and propagated the changes through to the end of the proof}
Since $\epsilon\le1/3$, one can easily check that
\[
	  \left(1-\epsilon\right)^{-1}\le1+2\epsilon ~~~\text{and}~~~
	\frac{1+\epsilon}{1-\epsilon}\le1+3\epsilon.
\]
 Consequently, we obtain that
\begin{equation}
	\left\Vert x-\widehat{x}\right\Vert \le\left(1+3\epsilon\right)\left\Vert x-x^{*}\right\Vert +\left(2+4\epsilon\right)\left\Vert \noise\right\Vert ,\qquad\mbox{when }\widehat{x}\in M_{\epsilon/40}.\label{eq:zhat bnd case 2}
\end{equation}
\notechecked{revised here through end of proof to account for new theorem statement} Combining \eqref{eq:zhat bnd case 1} and \eqref{eq:zhat bnd case 2}, we overall obtain that
\begin{align}\label{eq:temporary 1}
	\left\Vert x-\widehat{x}\right\Vert &
	 \le\max\left(\frac{\epsilon\condition}{40}\,,\,\left(1+3\epsilon\right)\left\Vert x-x^{*}\right\Vert +\left(2+4\epsilon\right)\left\Vert \noise\right\Vert \right)\nonumber\\
	 &\le \left(1+3\epsilon\right)\left\Vert x-x^{*}\right\Vert +\frac{\epsilon\condition}{40} +\left(2+4\epsilon\right)\left\Vert \noise\right\Vert,
\end{align}
which, to emphasize, is valid under the assumptions of Lemma \ref{lemma:stable embed of U(M,z)} and except for a probability of at most $2\failprob$.

On the other hand, according to Theorem~\ref{thm:kappa} and the remarks that followed it (see \eqref{eq:kappabound}), it holds that
\begin{equation}\label{eq:temporary}
\|x-\widehat{x}\| \le (1+2\epsilon)\left(2\sqrt{\frac{N}{M}}+5\right)\|x-x^*\|+(2+4\epsilon)\|n\|,
\end{equation}
except for a probability of at most $2\rho$ and as long as both \eqref{eq:mild assump on volume 1 main} and \eqref{eq:mmeasmain} hold.
 From \eqref{eq:temporary 1} and \eqref{eq:temporary}, we conclude that
\[
\left\Vert x-\widehat{x}\right\Vert \le
\min\left(  \left(1+3\epsilon\right)\left\Vert x-x^{*}\right\Vert +\frac{\epsilon\condition}{40}\,,\,  (1+2\epsilon)\left(2\sqrt{\frac{N}{M}}+5\right)\|x-x^*\| \right) +\left(2+4\epsilon\right)\left\Vert \noise\right\Vert,
\]
except for a probability of at most $4\rho$ and as long as both \eqref{eq:mmeasmain} and \eqref{eq:conv assump for par est} hold. This completes the proof of Theorem~\ref{theo:bound3}.


\section{Proof of Theorem~\ref{theo:bound4}}
\label{app:bound4}

Using the triangle inequality and (\ref{eq:bound3}), we have\notechecked{updated next inequality}
\begin{align}
\norm{\widehat{x} - x^\ast} & \le \norm{x-\widehat{x}} +
\norm{x-x^\ast} \nonumber\\
& \le  \min\left(  \left(2+3\epsilon\right)\left\Vert x-x^{*}\right\Vert +\frac{\epsilon\condition}{40}\,,\,  \left((2+4\epsilon)\sqrt{\frac{N}{M}}+6+10\epsilon\right)\|x-x^*\| \right)  +(2+4\eps)\norm{\noise}.
\label{eq:bound4a}
\end{align}
Now, since both $\widehat{x}$ and $x^\ast$ belong to $\manifold$, we
can invoke Lemma~\ref{lemma:d & dM} from~the Toolbox, which states that if
$\norm{\widehat{x} - x^\ast} \le \condition/2$, then
\begin{equation}
\gdist(\widehat{x}, x^\ast) \le \condition -
\condition\sqrt{1-2\norm{\widehat{x} -
x^\ast}/\condition}.\label{eq:bound4b}
\end{equation}
 To apply
this lemma, it is sufficient to know that\notechecked{updated next inequalities}
$$
 (2 +3
\eps)\norm{x-x^\ast}  +\frac{\eps \condition}{40}+(2+4\eps)\norm{\noise} \le \condition/2,
$$
i.e., that
$$
 \norm{x-x^\ast} + \frac{1+2\eps}{1+3\eps/2}\norm{\noise}\le
\frac{\condition}{4} \left(\frac{1 - \eps/20
}{1+3\eps/2}\right).
$$
For the sake of neatness, we may tighten this condition to\notechecked{edited next inequality}
$\norm{x-x^\ast} +\frac{10}{9}\norm{\noise}\le 0.163\condition$, which
implies the sufficient condition above (since $\eps \le 1/3$). Thus, if
$\norm{x-x^\ast}$ and $\norm{\noise}$ are sufficiently small (on the
order of the condition number $\condition$), then we may combine
(\ref{eq:bound4a}) and (\ref{eq:bound4b}), giving\notechecked{edited next inequality}
\begin{align}
& \gdist(\widehat{x}, x^\ast)\nonumber\\
 &\le \condition -
\condition\sqrt{1-\frac{2}{\condition}
\left(\min\left(  \left(2+3\epsilon\right)\left\Vert x-x^{*}\right\Vert +\frac{\epsilon\condition}{40}\,,\,  \left((2+4\epsilon)\sqrt{\frac{N}{M}}+6+10\epsilon\right)\|x-x^*\| \right)  +(2+4\eps)\norm{\noise}\right)} \nonumber \\
&= \condition - \condition\sqrt{1-\min\left( \frac{4+6\epsilon}{\tau}\|x-x^*\|+\frac{\epsilon}{20}\,,\,\tau^{-1}\left((4+8\epsilon)\sqrt{\frac{N}{M}}+12+20\epsilon\right)\|x-x^*\| \right) - \frac{4+8\epsilon}{\tau}\|n\|}. \label{eq:bound4c}
\end{align}
Under the assumption that \notechecked{edited next inequality} $\norm{x-x^\ast} +\frac{10}{9}\norm{\noise}\le
0.163\condition$, it follows that \notechecked{edited rest of sentence} the term inside the square root in the last line above must be nonnegative, and therefore \eqref{eq:bound4} holds.
%
%
This completes the proof of Theorem~\ref{theo:bound4}.

\footnotesize
\bibliographystyle{plain}
\bibliography{Revision-v10}

\end{document}